\documentclass[final]{siamltex}

\usepackage{graphics} 
\usepackage{epsfig} 
\usepackage{amsmath} 
\usepackage{amssymb}  
\usepackage{upgreek}
\usepackage{graphicx,psfrag, color}
\usepackage{cite}
\usepackage{latexsym}
\usepackage{amssymb}
\usepackage{graphics}
\usepackage{pstricks}
\usepackage{dsfont}
\usepackage{mathrsfs}
\usepackage{subfigure}
\usepackage{bbm}
\usepackage{latexsym}
\usepackage{lettrine}%
\usepackage{dsfont}
\usepackage[T1]{fontenc}
\usepackage[english]{babel}
\usepackage[latin1]{inputenc}
\usepackage{graphicx}
\usepackage{pgf, tikz, pgflibraryshapes}
\usepackage{amsmath,amssymb}
\usepackage{colortbl}
\usepackage{multicol}
\usepackage{stackrel}
\usepackage{colortbl}
    \definecolor{gray1}{rgb}{0,0,.8}
    \definecolor{gray2}{rgb}{0,1,0}
    \definecolor{gray3}{rgb}{0.8,0,.2}
\usepackage[pdftitle={},pdfauthor={Chiara
Ravazzi},pdfsubject={},pdfkeywords={},hyperfigures=true,colorlinks=true,citecolor=gray1,pagecolor=gray2,linkcolor=gray3,bookmarks=true,
bookmarksopen=true,bookmarksopenlevel=1,bookmarksnumbered=true,dvips=true,a4paper=true]{hyperref}

\newcommand{\e}{\mathrm{e}}
\newcommand{\E}{\mathbb{E}}
\renewcommand{\P}{\mathbb{P}}

\newcommand{\R}{\mathbb{R}}
\newcommand{\N}{\mathbb{N}}
\newcommand{\argmax}[1]{\underset{#1}{\mathrm{argmax\,}}}

\title{A distributed classification/estimation algorithm for sensor networks\thanks{A
preliminary version of some of the results has appeared in the proceedings of the 50st IEEE Conference on Decision and Control and European Control Conference, Orlando, Florida, 12-15 December 2011.}}
\author{Fabio Fagnani\thanks{DISMA (Dipartimento di Scienze Matematiche), Politecnico di Torino, Corso Duca degli Abruzzi, 24, I-10129 TO,
        (e-mail: {\tt\small fabio.fagnani@polito.it}}) \and Sophie M. Fosson\thanks{DET (Dipartimento di Elettronica e Telecomunicazioni),
Politecnico di Torino, Corso Duca degli Abruzzi, 24, I-10129 TO (e-mail: {\tt sophie.fosson@polito.it)} } \and Chiara Ravazzi\thanks{DET (Dipartimento di Elettronica e Telecomunicazioni),
Politecnico di Torino, Corso Duca degli Abruzzi, 24, I-10129 TO (e-mail: {\tt chiara.ravazzi@polito.it)} }}

\begin{document}
\maketitle

\begin{abstract}
In this paper, we address the problem of simultaneous classification and estimation of hidden parameters in a sensor network with communications constraints.
In particular, we consider a network of noisy sensors which measure a common scalar unknown parameter. We assume that a fraction of the nodes represent faulty sensors, whose measurements are poorly reliable. The goal for each node is to simultaneously identify its class (faulty or non-faulty) and estimate the common parameter. 

We propose a novel cooperative iterative algorithm which copes with the communication constraints imposed by the network and shows remarkable performance. 
Our main result is a rigorous proof of the convergence of the algorithm and a characterization of the limit behavior. We also show that, in the limit when the number of sensors goes to infinity,  the common unknown parameter is estimated with arbitrary small error, while the classification error converges to that of the optimal centralized maximum likelihood estimator.
We also show numerical results that validate the theoretical analysis and support their possible generalization. We compare our strategy with the Expectation-Maximization algorithm and we discuss trade-offs in terms of robustness, speed of convergence and implementation simplicity.
\end{abstract}

%
%

\section{Introduction} 
Sensor networks are one of the most important technologies introduced in our century. Promoted by the advances in wireless
communications and by the pervasive diffusion of smart sensors, wireless sensor networks are largely used nowadays for a
variety of purposes, e.g., environmental and habitat surveillance, health and security monitoring, localization, targeting, event
detection.

A sensor network basically consists in the deployment of a large numbers of small devices, called sensors, that have the ability to
perform measurements and simple computations, to store few amounts of data, and to communicate with other devices. In this paper, we
focus on \emph{ad hoc} networks, in which  communication is local: each sensor is connected only with a
restricted number of other sensors. This kind of cooperation allows to perform
elaborate operations in a self-organized way, with no centralized supervision or data fusion center, with a substantial energy and
economic saving on  processors and communication links. This allows to construct large sensor networks  at contained cost.

A problem that can be addressed through ad hoc sensor networking is the distributed estimation: given an unknown physical parameter
(e.g., the temperature in a room, the position of an object), one aims at estimating it using the sensing capabilities of a network. Each sensor  performs a (not exact) measurement and shares it with the sensors with which it can establish a communication; in turn, it receives information and consequently updates its own estimate. If the network is connected, by iterating the sharing procedure, the
information propagates and a consensus can be reached. Neither centralized coordinator nor data fusion center is present.  The
mathematical model of this problem must envisage the presence of noise in measurements, which are naturally corrupted by
inaccuracies, and possible constraints on the network  in terms of communication, energy or bandwidth limitations, and of  necessity of
quantization or data compression. 

\emph{Distributed estimation} in ad hoc sensor networks has been widely studied in the literature. For the problem of estimating an unknown common parameter, typical approach is to consider distributed versions of classical maximum likelihood (ML) or maximum-a-posteriori (MAP) estimators. Decentralization can be obtained, for instance, through consensus type protocols (see \cite{DBLP:journals/siamco/HuangM09}, \cite{olf05}, \cite{olf09}) adapted to the communication graph of the network, or by belief propagation methods  \cite{moa06} and \cite{saol06}.


A second important issue is \emph{sensors' classification}, which we define as follows \cite{Duda2001}. Let us imagine that sensors
can be divided into different classes according to peculiar properties, e.g., measurements' or processing capabilities, and that no
sensor knows to which class it belongs: by classification, we then intend the labeling procedure that each sensor undertakes to determine  its affiliation. This task is addressed to a variety of clustering purposes, for example, to rebalance the computation load in a
network where sensors can be distinguished according to their processing power. 
On most occasions, sensors' classification is faced through some distributed estimation, the underlying idea being the following: each
sensor performs its measurement of a parameter, then iteratively modifies it on the basis of information it receives; during this
iterative procedure the sensor learns something about itself which makes it able to estimate its own configuration.

In this paper, we consider the following model: each sensor $i$ performs a measurement $y_i=\theta^{\star}+\omega_i^{\star}\eta_i$, where
$\theta^{\star}\in\R$ is the unknown global parameter, $\omega_i^{\star}>0$ is the unknown status of the sensor, and $\eta_i$ is a Gaussian
random noise. The more $\omega_i^{\star}$ is large, the more the sensor $i$ is malfunctioning, that is, the quality if its measurement is
low. The $\omega_i^{\star}$ parameter is supposed to belong to a discrete set, in particular in this paper we consider the binary case.

The goal of each unit $i$ is to estimate the parameter $\theta^{\star}$ and the specific configuration $\omega^{\star}_i$. The presence of the common unknown parameter $\theta^{\star}$ imposes a coupling between the different nodes and makes the problem interesting. 

An additive version of the
aforementioned model has been studied in 
\cite{CFSZ11}, where measurement is given by $y_i=\theta^{\star}+\omega_i^{\star}+\eta_i$. Another related problem is the
so-called calibration problem \cite{dog06,calibration11}: sensor $i$ performs a noisy linear measurement  $y_i=A_i \theta+\eta_i$  where the unknown $\theta$ and 
$A_i$ are a vector and a matrix, respectively,  while $\eta_i$ is a noise; the goal consists in the estimation of $\theta$ and of
$A_i$, the latter being known as calibration problem. 

All these  are particular cases of the problem of the estimation of Gaussian mixtures'
parameters \cite{tit85, SIAMRev}. This perspective has been studied for sensor networks in \cite{now03}, \cite{now04}, \cite{gu08}, and \cite{das09} where distributed
versions of the Expectation-Maximization (EM) algorithm have been proposed.
A network is given where each node independently performs the E-step through local observations. In particular, in \cite{gu08} a consensus filter is used to propagate the local information. The tricky point of such techniques is the choice of the number of averaging iterations between two consecutive M-steps, which must be sufficient to reach consensus.

The aim of this paper is the development of a distributed, iterative procedure which copes with the communication constraints imposed by the network and computes an estimation ($\widehat\theta, \widehat \omega$) approximating the maximum likelihood optimal solution of the proposed problem. 
The core of our methodology is an  Input Driven Consensus Algorithm (IA for short), introduced in \cite{Fagnani_Fosson_Ravazzi_2011}, which takes care of the estimation of the parameter $\theta^*$. IA is coupled with a classification step where nodes update the estimation of their own type $\omega^*_i$ by a simple threshold estimator based on the current estimation of $\theta^*$. The fact of using a consensus protocol working on inputs instead, as more common, on initial conditions, is a key strategic fact: it serves the purpose of  using the innovation coming from the units who are modifying the estimation of their status, as time passes by. Our main theoretical contribution is a  complete analysis of the algorithm
 in terms of convergence and of behaviour with respect to the size of the network. With respect to other approaches like distributed EM for which convergence results are missing, this makes an important difference. We also present a number of numerical simulations showing the remarkable performance of the algorithm which, in many situations, outperform classical choices like EM.

The outline of the paper is the following. In Section \ref{sec1} we shortly present some graph nomenclature needed in the paper. Section \ref{problemstatement} is devoted to a formal description of the problem and to a discussion of the classical centralized maximum likelihood solution.  In Section \ref{ouralgorithm}, we present the details and the analysis of  our IA. Our main results are Theorems \ref{teo:convergenza} and \ref{concentration2}: Theorem \ref{teo:convergenza} ensures that, under suitable assumptions on the graph, the algorithm converges to a local maximum of the log-likelihood function; Theorem \ref{concentration2} is a concentration result establishing that when the number of nodes $N\to +\infty$, the estimate $\widehat \theta$ converges to the true value $\theta^*$ (a sort of asymptotic consistency). Finally, we also study the behavior of the discrete estimate $\widehat \omega$ by analyzing the performance index the relative classification
error over the network when $N\to +\infty$ (see Corollary \ref{concentration4}).
Section \ref{simulations} contains a set of numerical simulations carried on different graph architectures: complete, circulant, grids, and random geometric graphs. Comparisons are proposed with respect to the optimal centralized ML solution and also with respect to the EM solution. Finally, a long Appendix contains all the proofs.

\section{General notation and graph theoretical preliminaries}\label{sec1} Throughout this paper, we use the following notational convention. 
We denote vectors with small letters, and matrices with capital letters. Given a matrix $M$, $M^T$ denotes its transpose. Given a vector $v$, $||v||$ denotes its Euclidean norm. $\mathbf{1}_A$ is the indicator function of set $A$.
Given a finite set $\mathcal{V}$, $R^{\mathcal{V}}$ denotes the space of real vectors with components labelled  by elements of $\mathcal{V}$.
Given two vectors $x,z\in\R^{\mathcal V}$, $\mathrm{d_H}(x,z)=|\{i\in\mathcal V: x_i\neq z_i\}|$. 
We use the convention that a summation over an empty set of indices is equal to zero, while a product over an empty set gives one. 

A symmetric graph is a pair $\mathcal{G}=(\mathcal{V, E})$ where $\mathcal{V}$ is a set, called the set of vertices, and $\mathcal E\subseteq \mathcal{V\times V}$ is the set of edges with the property that $(i,i)\not\in \mathcal E$ for all $i\in\mathcal V$ and $(i,j)\in \mathcal E$ implies $(j,i)\in \mathcal E$. 
$\mathcal{G}$ is strongly connected if, for all $i,j\in \mathcal V$, there exist vertices $i_1,\dots i_s$ such that $(i,i_1), (i_1,i_2),\dots ,(i_s,j)\in \mathcal E$. 
To any symmetric matrix $P\in\R^{\mathcal V\times\mathcal V}$ with non-negative elements, we can associate a graph $\mathcal G_P=(\mathcal V, \mathcal{E}_P)$ by putting $(i,j)\in \mathcal{E}_P$ if and only if $P_{ij}>0$. $P$ is said to be adapted to a graph $\mathcal G$ if $\mathcal G_P\subseteq \mathcal G$. A matrix with non-negative elements $P$ is said to be stochastic if $\sum_{j\in \mathcal V}P_{ij}=1$ for every $i\in\mathcal V$. Equivalently, denoting by ${\mathbbm 1}$ the vector of all $1$ in $\R^{\mathcal V}$, $P$ is stochastic if $P{\mathbbm 1}={\mathbbm 1}$. $P$ is said to be primitive if there exists $n_0\in\N$ such that $P^{n_0}_{ij}>0$ for every $i,j\in\mathcal V$. A sufficient condition ensuring primitivity is that $\mathcal G_P$ is strongly connected and $P_{ii}>0$ for some $i\in \mathcal V$. 

\section{Bayesian modeling for estimation and classification}\label{problemstatement}

\subsection{The model}

In our model, we consider a network, represented by a symmetric graph $\mathcal{G} = (\mathcal{V}, \mathcal{E})$. $\mathcal{G}$ represents the system communication architecture. We denote the number of nodes by  $N=|\mathcal{V}| $. We assume that each node $i\in\mathcal{ V}$ measures the observable \begin{equation}\label{model}y_i=\theta^{\star} +\omega^{\star}_i\eta_i\end{equation} where 
$\theta^{\star}\in\mathbb{R}$ is an unknown parameter, $\eta_i$'s
Gaussian noises $\mathsf{N}(0,1)$, $\omega^{\star}_i$'s Bernoulli random variables taking values in $\{\alpha, \beta\}$ (with ${\mathbb P}(\omega^{\star}_i=\beta)=p$). We assume all the random variables $\eta_i$'s and $\omega^{\star}_i$'s to be mutually independent. Notice that each $y_i\in\mathbb{R}$ is a Gaussian mixture distributed according to the probability density function 
\begin{gather}\label{gaumix}
f(y_i)=(1-p) f(y_i|\theta^{\star},\alpha)+p f(y_i|\theta^{\star},\beta)\\
f(y_i|\theta^{\star},x)=\frac{1}{x\sqrt{2\pi }}\e^{-\frac{(y_i-\theta^{\star})^2}{2 x^2}}\quad x\in\{\alpha,\beta\}.
\end{gather}
The binary model of $\omega^{\star}$ is motivated by different scenarios: as an example, if $0<\alpha << \beta$, the nodes of type $\beta$ may represent a subset of  faulty sensors, whose measurements are poorly reliable; the aim may be the detection of faulty sensors in order to switch them off or neglect their measurements, or for other clustering purposes. It is also realistic to assume that some a-priori information about the quantity of faulty sensors is extracted, e.g., from experimental data on the network, and it is conceivable to represent such information as an a-priori distribution. This is why we assume a Bernoulli distribution on each $\omega^{\star}_i$; on the other hand, we suppose that no a-priori information is available on the unknown parameter $\theta^{\star}$. However, the addition of an a priori probability distribution on $\theta^*$ does not significantly alter our analysis and our results.



\subsection{The maximum likelihood solution}
The goal is to estimate the parameter $\theta^{\star}$ and the specific configuration $\omega^{\star}_i$ of each unit. Disregarding the network constraints, a natural solution to our problem would be to consider  a joint ML in $\theta^{\star}$ and MAP in the $\omega^{\star}_i$'s (see \cite{yer00, bb07}). 
Let $f(y,\omega|\theta)$ be the joint distribution of  $y$ and $\omega$ (density in $y$ and probability in $\omega$) given the parameter $\theta$, and consider the rescaled log-likelihood function 
\begin{align}\begin{split}\label{likelihood} L_N(\theta,\omega)&:=\frac{1}{N}\log f(y,\omega|\theta).
\end{split}
\end{align}
The hybrid ML/MAP solution, which for simplicity for now on we will refer to as the ML solution, prescribes to choose $\theta$ and
$\omega$ which maximize $L_N(\theta,\omega)$ 
\begin{equation}\label{optproblem}
(\widehat{\theta}^{\mathrm{ML}}, \widehat{\omega}^{\mathrm{ML}}):=\argmax{\theta \in \mathbb{R},~ \omega\in\{\alpha,\beta\}^{\mathcal{V}}}L_N(\theta,\omega).
\end{equation}
Standard calculations lead us to
\begin{equation}\label{LN} L_N(\theta,\omega)=-\frac{1}{N}\sum_{j\in\mathcal{V}}\left(\frac{(y_j-\theta)^2}{2\beta^2}+ {\mathbf 1}_{\{\omega_j=\alpha\}}\left(\frac{(y_j-\theta)^2}{2}\left(\frac{1}{\alpha^2}-\frac{1}{\beta^2}\right)+\log\frac{1-p}{p}\frac{\beta}{\alpha}\right)\right)+c\end{equation}
where $c$ is a constant.  It can be noted that partial maximizations of $L_N(\theta,\omega)$ with respect to just one of the two variables have simple representation. 
Let
\begin{equation}\widehat{\theta}(\omega):=\argmax{\theta}{L}_N(\theta,\omega)\qquad \widehat{\omega}(\theta):=\argmax{\omega}L_N(\theta, \omega).
\end{equation}
Then
\begin{equation} \label{partial max}\widehat{\theta}(\omega) =\frac{\sum_jy_j/\omega_j^{2}}{\sum_j1/\omega_j^{2}}\qquad 
 \widehat{\omega}(\theta)_i=
\begin{cases}
\alpha&\text{if }|y_i-{{\theta}}|<\delta\\
\beta &\text{otherwise}
\end{cases}
\end{equation} where $$\delta=\sqrt{2\frac{\ln\left(\frac{1-p}{p}\frac{\beta}{\alpha}\right)}{\frac{1}{\alpha^2}-\frac{1}{\beta^2}}}.$$
The ML solution can then be obtained, for instance, by considering
\begin{equation}\label{ML sol}\widehat{\theta}^{\mathrm{ML}}=\argmax{\theta}L(\theta, \widehat{\omega}(\theta))\,,\quad \widehat{\omega}^{\mathrm{ML}}=\widehat{\omega}(\widehat{\theta}^{\mathrm{ML}}).\end{equation}
It should be noted how the computation of the $(\widehat{\omega}^{\mathrm{ML}})_i$'s becomes totally decentralized once $\widehat{\theta}^{\mathrm{ML}}$ has been computed. For the computation of $\widehat{\theta}^{\mathrm{ML}}$ instead one needs to gather information from all units to compute $L_N(\theta, \widehat \omega(\theta))$ and it is not at all evident how this can be done in a decentralized way. Moreover, further difficulties are caused by the fact that $L_N(\theta, \widehat \omega(\theta))$ may contain many local maxima, as shown in Figure \ref{L2}.

It should be noted that $L_N(\theta, \widehat{\omega}(\theta))$ is differentiable except at a finite number of points, and between two successive non-differentiable points  the function is concave. Therefore, the local maxima of the function coincide with its critical points. On the other hand, the derivative, where it exists, is given by
\begin{equation}\label{derivative} \begin{split}                                   
\frac{d}{d\theta}L_N(\theta, \widehat{\omega}(\theta))&=\left(\frac{1}{\beta^2}-\frac{1}{\alpha^2}\right)\frac{1}{N}\sum_{i\in\mathcal V}(\theta-y_i){\mathbf 1}_{\{|y_i-\theta|<\delta\}}-\frac{1}{\beta^2}\left(\theta-\frac{1}{N}\sum_{i\in\mathcal V}y_i\right).
\end{split}\end{equation}
Stationary points can therefore be represented by the relation
\begin{equation}\label{stationary1}\theta
=\frac{\frac{1}{\beta^2}\sum_iy_i+\left(\frac{1}{\alpha^2}-\frac{1}{\beta^2}\right)\sum_iy_i{\mathbf
1}_{\{|y_i-\theta|<\delta\}}}{N\frac{1}{\beta^2}+\sum_i{\mathbf
1}_{\{|y_i-\theta|<\delta\}}\left(\frac{1}{\alpha^2}-\frac{1}{\beta^2}\right)}.\end{equation}
A moment of thought shows us that (\ref{stationary1}) is equivalent to the relation $\theta =\widehat\theta(\widehat \omega(\theta))$. 

This representation will play a key role in the sequel of this paper.

\begin{figure}[h!]\includegraphics[width=0.5\columnwidth]{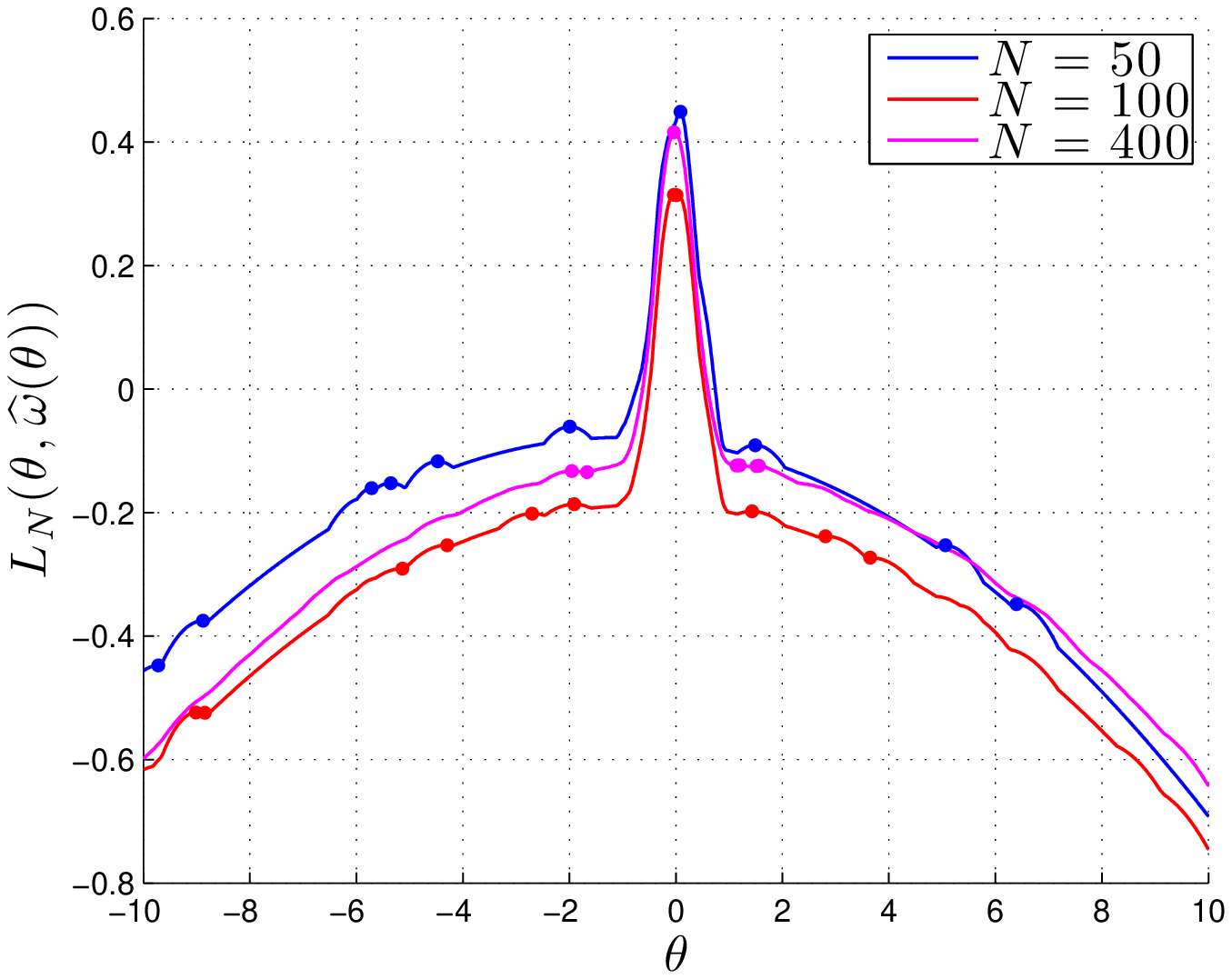}
\includegraphics[width=0.5\columnwidth]{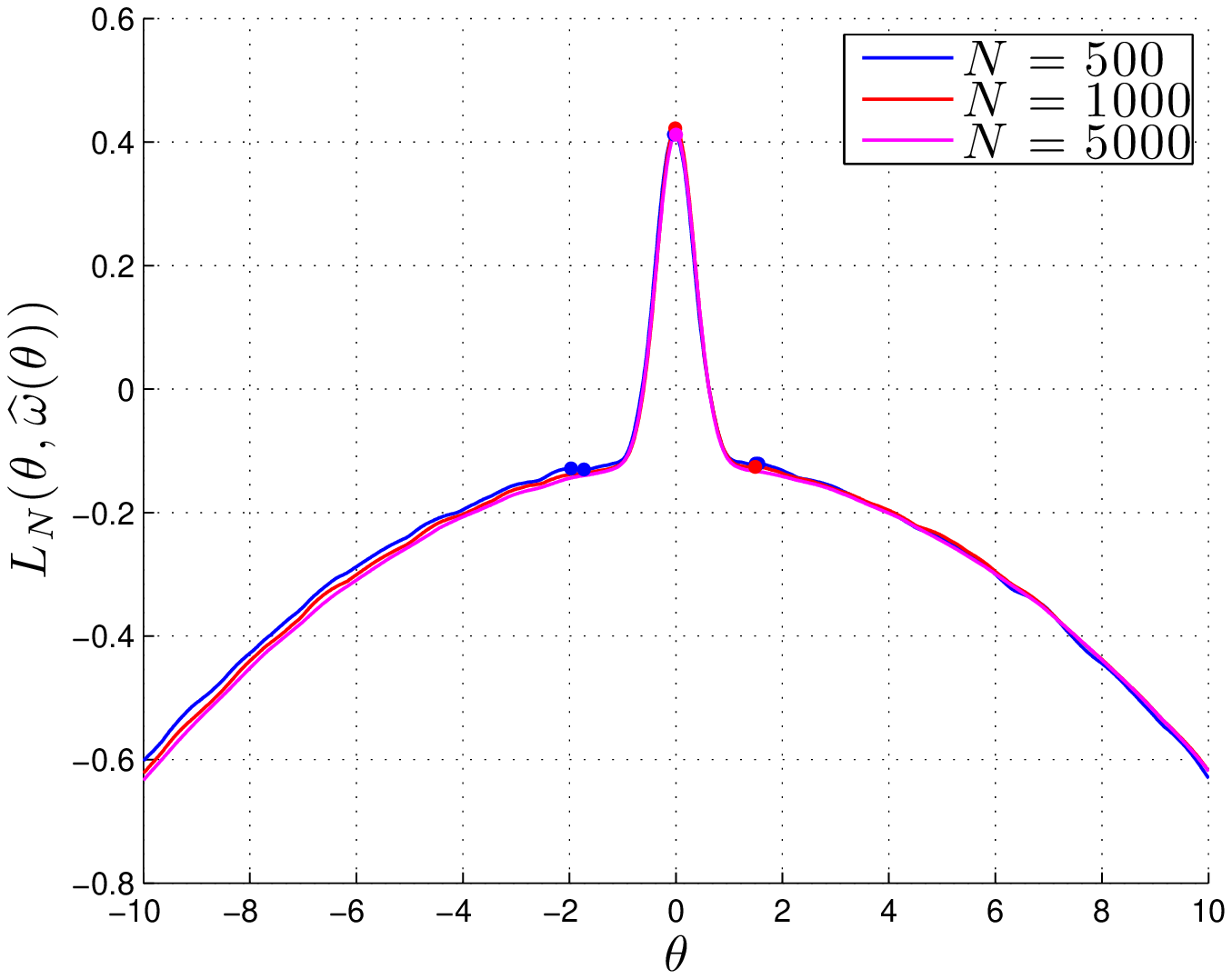}
\begin{center}
\caption{$\alpha=0.3, \beta=10, p=0.25$: Plot of function $L_N(\theta,\widehat{\omega}(\theta))$ as a function of $\theta$ and size $N\in\{50,100,400,500,1000,5000\}$.} \label{L2} \label{grafici_L}
\end{center}
\end{figure}

\subsection{Iterative centralized algorithms}\label{iterative}
The computational complexity of the optimization problem (\ref{optproblem}) is practically unfeasible in most situations. 
However, relations \eqref{partial max} suggest a simple way to construct an iterative approximation of the ML solution (which we will denote IML). 
The formal pattern is the following: fixed $\widehat{\omega}^{(0)}=\alpha \mathbbm{1}$, for $t=0,1,\dots$, we consider the dynamical system
\begin{gather*}
\widehat{\theta}^{(t+1)}=\frac{\sum_{j=1}^N y_j\left[\widehat{\omega}_j^{(t)}\right]^{-2}}{\sum_{j=1}^N\left[\widehat{\omega}_j^{(t)}\right]^{-2}}\\
 \widehat{\omega}(\theta)^{(t+1)}_i=
\begin{cases}
\alpha&\text{if }|y_i-{{\theta}}|<\delta\\
\beta &\text{otherwise}
\end{cases}~~~\text{ for any } i=1,\dots,N.
\end{gather*}
The algorithm stops whenever $|\widehat{\theta}^{(t+1)}-\widehat{\theta}^{(t)}|<\varepsilon$, for some fixed tolerance $\varepsilon>0$.

A more refined iterative solution is given by the so-called Expectation-Maximization (EM) algorithm \cite{dem77}.  The main idea is to  introduce a hidden (say, unknown and unobserved) random variable in the likelihood; then, at each step, one computes the mean of the likelihood function with respect to the hidden variable and finds its maximum. Such a method seeks to find the maximum likelihood solution, which in many cases cannot be formulated in a closed form.  EM is widely and successfully used in many frameworks and in principle it could also be applied to our problem. In our context, making the variable $\omega$ to play the part of the hidden variable, equations for EM become (see the tutorial \cite{bilmes1997-em} for their derivation)

Given $\widehat{\theta}^{(0)}\in\mathbb{R}$, for $t=0,1,\dots$,
\begin{enumerate}
\item E-step: for all node $i\in\mathcal V$,
  \begin{equation*}
  q_i^{(t)}=\mathbb{P}\left(\widehat{\omega}_i^{(t)}=\alpha|y,\widehat{\theta}^{(t)}\right)=\frac{(1-p)f\left(y|\widehat{\omega}_i^{(t)}=\alpha,\widehat{\theta}^{(t)}\right)}{(1-p)f\left(y|\widehat{\omega}_i^{(t)}=\alpha,\widehat{\theta}^{(t)}\right)+pf\left(y|\widehat{\omega}_i^{(t)}=\beta,\widehat{\theta}^{(t)}\right)}.  
  \end{equation*}
\item M-step: 
  \begin{equation*}
    \widehat{\theta}^{(t+1)}=\frac{\sum_{j\in\mathcal{V}}y_j\left(q_j^{(t)}\alpha^{-2}+(1-q_j^{(t)})\beta^{-2}\right) }{\sum_{j\in\mathcal{V}}q_j^{(t)}\alpha^{-2}+(1-q_j^{(t)})\beta^{-2}}.
  \end{equation*}
\end{enumerate}
The algorithm stops whenever $|\widehat{\theta}^{(t+1)}-\widehat{\theta}^{(t)}|<\varepsilon$, for some fixed tolerance $\varepsilon>0$. It is worth to notice that $q_i^{(t)}$ computed in the E-step actually is the expectation of the binary random variable $\mathbf{1}_{\{\widehat{\omega}_i^{(t)}=\alpha\}}$. On the other hand $\widehat{\theta}^{(t+1)}$ computed in the M-step is the maximum of such expectation.

An important feature of EM is that it is possible to prove the convergence of the sequence $\{\widehat{\theta}^{(t)}\}_{\in\mathbb{N}}$ to a local maximum of the expected value of the log-likelihood with respect to the unknown data $\omega$, a result which is instead not directly available for IML. 
Both algorithms however share the drawback of requiring centralization. Distributed versions of the EM have been proposed (see, e.g., \cite{now03}, \cite{gu08}) but convergence is not guaranteed for them. In Section \ref{simulations} we will compare both these algorithms against the distributed IA we are going to present in the next section. While it is true that EM always outperforms IML, algorithm IA outperforms both of them for small size algorithms, while shows comparable performance to EM for large networks.

\section{Input driven consensus algorithm} \label{ouralgorithm}

\subsection{Description of the algorithm}

In this section we propose a distributed iterative algorithm approximating the centralized ML estimator. The algorithm is suggested by
the expressions in \eqref{partial max} and consists of the iteration of two steps: an averaging step where all units aim at computing $\widehat\theta$ through a sort of Input Driven Consensus Algorithm (IA) followed by an update of the classification estimation performed autonomously by all units.


Formally, IA is parametrized by a symmetric stochastic matrix $P$, adapted to the communication graph
$\mathcal{G}$ ($
P_{ij}>0
$ if and only if, $(i,j)\in\mathcal{E}$), and by a real sequence $\gamma^{(t)}\rightarrow 0$. 
Every node $i$ has three messages stored in its memory at time $t$, denoted with $\mu_i^{(t)},\nu_i^{(t)}$, and $\widehat \omega^{(t)}_i$. Given the initial conditions $\mu_i^{(0)}=0,\nu_i^{(0)}=0$ and the initial estimate $\widehat{\omega}_i^{(0)}=\alpha$, the dynamics consists of the following steps.
\begin{enumerate}
\item \textit{Average step}: \begin{subequations}\label{munu}\begin{align}\label{mu}
\mu_i^{(t+1)}&=(1-\gamma^{(t)}){\sum_{j}P_{ij}\mu_j^{(t)}}+\gamma^{(t)}{{y_i}{\left(\widehat{\omega}_i^{(t)}\right)^{-2}}}
\\ \label{nu}
\nu_i^{(t+1)}&=(1-\gamma^{(t)}){\sum_{j}P_{ij}\nu_j^{(t)}}+\gamma^{(t)}{{\left(\widehat{\omega}_i^{(t)}\right)^{-2}}}\\ \label{sys1}
\widehat{\theta}^{(t+1)}_i&={\mu_i^{(t+1)}}/{\nu_i^{(t+1)}}.
\end{align}
\end{subequations}
\item \textit{Classification step}:
\begin{equation}\label{sys}
\widehat{\omega}_i^{(t+1)}=\widehat \omega_i(\widehat\theta^{(t+1)})=\left\{\begin{array}{cl}
 \alpha & \text{if }|y_i-\widehat{\theta}_i^{(t+1)}|<\delta\\
 \beta&\text{otherwise.}
 \end{array}\right.
 \end{equation}
 \end{enumerate}
It should be noted that the algorithm provides a distributed protocol: each node only needs to be aware of its neighbours and no further information about the network topology is required.

\subsection{Convergence}
The following theorem ensures the convergence of IA. The proof is rather technical and therefore deferred to Appendix \ref{appA}.

\begin{theorem}\label{teo:convergenza}
Let 
\begin{enumerate}
\item[(a)] $\gamma^{(t)}\rightarrow 0$, $\gamma^{(t)}\geq 1/t$, and
$\gamma^{(t)}=\gamma^{(t+1)}+o(\gamma^{(t+1)})$ for $t\to
+\infty$;
\item[(b)] $P\in \mathbb{R}_+^{\mathcal{V}\times \mathcal{V}}$ be a
stochastic, symmetric, and primitive
matrix with positive eigenvalues. \end{enumerate}
Then, there exist $\widehat{\omega}^{{IA}}\in\{\alpha,\beta\}^{\mathcal{V}}$ and $\widehat{\theta}^{{IA}}\in{\mathbb R}$ such that
\begin{enumerate}
\item 
$$
\lim_{t\rightarrow+\infty} \widehat{\omega}^{(t)}\stackrel{\text{a.s.}}{=}\widehat{\omega}^{IA}\,,\qquad \lim_{t\rightarrow+\infty} \widehat{\theta}^{(t)}_i \stackrel{\text{a.s.}}{=}
\widehat{\theta}^{IA}$$ for all $i\in\mathcal V$;
\item they satisfy the relations 
$$\widehat{\theta}^{IA}=\widehat{\theta}(\widehat \omega^{IA})\,,\  \widehat{\omega}^{IA}=\widehat{\omega}(\widehat \theta^{IA}).$$
\end{enumerate}
\end{theorem}
A number of remarks are in order. 
\begin{itemize}
\item The assumption on the eigenvalues of $P$ is essentially a technical one: in simulations it does not seem to have a crucial role, but we need it in our proof of convergence.
On the other hand, given any symmetric stochastic primitive $P$, we cam consider a 'lazy' version of it $P_{\tau}=(1-\tau)I+\tau P$ and notice that for $\tau\in (0,1)$ sufficiently small, indeed $P_{\tau}$ will satisfy the assumption on the eigenvalues.
\item The requirement $\gamma^{(t)}\geq 1/t$ is not new in decentralized algorithms (see for instance the Robbins-Monro algorithm, introduced in \cite{StochApprox}) and serves the need of maintaining 'active' the system input for sufficiently long time. Less classical is the assumption $\gamma^{(t)}\sim\gamma^{(t+1)}$ which is essentially a request of regularity in the decay of $\gamma^{(t)}$ to $0$. Possible choices of $\gamma^{(t)}$ satisfying the above conditions are $\gamma^{(t)}=t^{-\zeta}$ for $\zeta\in (0,1)$, or $\gamma^{(t)}=t^{-1}(\ln t)^\alpha$ for any $\alpha >0$.
\item The proof (see Appendix \ref{appA}) will also give an estimation on the speed of convergence: indeed it will be shown that $||\widehat{\theta}^{(t)}-\widehat{\theta}^{IA}||=O(\gamma^{(t)})$ for $t\rightarrow\infty$.
\item Relations in item 2. implies that $\widehat{\theta}^{IA}$ is a local maximum of the function $L_N(\theta,\widehat{\omega}(\theta))$ (see (\ref{stationary1})).
\end{itemize}


\subsection{Limit behavior}
In this section we present results on the behavior of our algorithm for $N\to +\infty$. All quantities derived so far are indeed function of network size $N$. In order to emphasize the role of $N$, we will add an index $N$ when dealing with quantities like $\theta^{\star}$ (e.g. $\widehat{\theta}^{\text{ML}}_N$). Instead we will not add anything to expressions where there are vectors $\omega$ involved since their dimension is itself $N$.

Figure \ref{grafici_L} shows a sort of concentration of the local maxima of $L_N(\theta,\widehat{\omega}(\theta))$ to a global maximum for large $N$. Considering that IA converges to a local maximum, this observation would lead to the conclusion that, for large $N$, the IA resembles the optimal ML solution. This section provides some results which make rigorous these considerations.

Notice first that, applying the uniform law of large numbers \cite{ULLN} to the expression (\ref{LN}), we obtain that, for any compact $K\subseteq \R$, almost surely
\begin{equation}\label{uniform}\lim\limits_{N\to +\infty}\max_{\theta\in K}\left|L_N(\theta,\widehat{\omega}(\theta))- \int_{\mathbb{R}} \mathcal{J}(s, \theta)f(s)\mathrm{d}s\right| =0\end{equation}
where
\begin{equation}
   \mathcal{J}(s, \theta)=-\left(\frac{(s-\theta)^2}{2\beta^2}+ {\mathbf 1}_{\{\omega_j=\alpha\}}\left(\frac{(s-\theta)^2}{2}\left(\frac{1}{\alpha^2}-\frac{1}{\beta^2}\right)+\log\frac{1-p}{p}\frac{\beta}{\alpha}\right)\right)+c
\end{equation}
where $c$ is the same constant as in (\ref{LN}).
The limit function $\int_{\mathbb{R}} \mathcal{J}(s, \theta)f(s)\mathrm{d}s$ turns out to be differentiable for every value of $\theta$ and to have a unique stationary point for $\theta=\theta^*$ which turns out to be the global minimum. Unfortunately, this fact by itself does not guarantee that global and local minima will indeed converge to $\theta^*$. 
In our derivations the properties of the function $\int_{\mathbb{R}} \mathcal{J}(s, \theta)f(s)\mathrm{d}s$ will not play any direct role and therefore they will not be proven here. The main technical result which will be proven in Appendix \ref{appB} is the following:

\begin{theorem}\label{concentration2} Denote by $\mathcal{S}_N$ the set of local maxima of $L(\theta,\widehat{\omega}(\theta))$. Then, 
\begin{equation}\label{local minima}
\lim\limits_{N\to +\infty}\max\limits_{\xi\in{\mathcal S}_N}|\xi-\theta^*| =0\end{equation}
almost surely and in mean square sense.
\end{theorem}

This has an immediate consequence,

\begin{corollary}\label{concentration3}
\begin{equation}\lim\limits_{N\to +\infty}\widehat{\theta}^{IA}_N=\lim\limits_{N\to +\infty}\widehat{\theta}^{\mathrm{ML}}_{N}=\theta^{\star}
\end{equation}
almost surely and in mean square sense.
\end{corollary}

Regarding the classification error, we have instead the following result:
\begin{proposition}\label{concentration4}
\begin{equation}\label{LB}\begin{split}
\lim_{N\rightarrow+\infty}\ \frac{1}{N}\E d_H(\widehat{\omega}^{IA}, \omega^{\star})&=\lim_{N\rightarrow+\infty}\ \frac{1}{N}\E d_H(\widehat{\omega}^{{\rm ML}}, \omega^{\star})\\&
=q(p,\alpha,\beta)\\
\end{split} \end{equation}
 where 
$$
q(p,\alpha,\beta)=(1-p)\mathrm{erfc}\left(\frac{\delta}{\alpha\sqrt{2}}\right)+p\left[1-\mathrm{erfc}\left(\frac{\delta}{\beta\sqrt{2}}\right)\right]
$$
and
$\mathrm{erfc}(x):=\frac{2}{\sqrt{\pi}}\int_{x}^{+\infty}\e^{-t^2}\mathrm{d}t$ is the complementary error function.
\end{proposition}

These results ensure that the IA performs, in the limit of large number of units $N$, as the centralized optimal ML estimator. Moreover, they also show, consistency in the estimation of the parameter $\theta^{\star}$. As expected, for $N\to +\infty$ the classification error does not go to $0$ since the increase of measurements is exactly matched by the same increase of variables to be estimated. Consistency however is obtained when 
$p$ goes to zero since we have that
$\lim_{p\rightarrow 0} q(p,\alpha,\beta)=0.$
Moreover, notice that the dependence of function $q $ on the parameters $\alpha$ and $\beta$ is exclusively through their ratio $\beta/\alpha$. In
particular, we have
$$
\lim_{\beta/\alpha\rightarrow+\infty} q(p,\alpha,\beta)=0\qquad\lim_{\beta/\alpha\rightarrow1} q(p,\alpha,\beta)=1.
$$

\section{Simulations}\label{simulations}

In this section, we propose some numerical simulations. We test our algorithm for different graph architectures and dimensions, and we compare it with the IML and
EM algorithms. Our goal is to give evidence of the theoretical results' validity and also to evaluate cases that are not
included in our analysis: the good numerical outcomes we obtain suggest that convergence should hold in broader frameworks. The numerical setting for our simulations is now presented.\\

\textbf{Model}: the sensors perform measurements according to the model \eqref{model} with $\theta^{\star}=0$, $\alpha=0.3$,
$\beta=10$; the prior probability $\mathbb{P}(\omega^{\star}_i=\beta)$ is equal to $p=0.25$.\\ 

\textbf{Communication architectures}: given a strongly connected symmetric graph $\mathcal G=({\mathcal V},\mathcal E)$, we use the so-called Metropolis random walk construction for $P$ (see \cite{Xiao06distributedaverage}) which amounts to the following: if $i\neq j$
$$P_{ij}=\left\{\begin{array}{ll}0 \quad &{\rm if}\, (i,j)\not\in \mathcal E\\ \left(\max\{\mathrm{deg}(i)+1, \mathrm{deg}(j)+1\}\right)^{-1} \quad &{\rm if}\, (i,j)\in \mathcal E\end{array}\right.$$
where $\mathrm{deg}(i)$ denotes the degree (the number of neighbors) of unit $i$ in the graph $\mathcal G$. $P$ constructed in this way is automatically irreducible and aperiodic.

We consider the following topologies: 
\begin{enumerate}
\item \textit{Complete graph}: $P_{ij}=\frac{1}{N}$ for every $i,j=1,\dots, N$; it actually corresponds to the centralized case.
%
\item \textit{Ring:}\label{case: circulant} $N$ agents are disposed on a circle, and each agent
communicates with its first neighbor on each side (left and right).  The corresponding circulant symmetric matrix $P$ is given by $P_{ij}=\frac{1}{3}$ for every $i=2,\dots,N-1$ and $j\in\{i-1,i,i+1\}$; $P_{11}=P_{12}=P_{1N}=\frac{1}{3}$; $P_{N1}=P_{NN-1}=P_{NN}=\frac{1}{3}$; $P_{ij}=0$ elesewhere.

\item \textit{Torus-grid graph}: sensors are deployed on a two dimensional grid and are
each connected with their four neighbors; the last node of each row of the grid is connected with the first node of the same row, and
analogously on columns, so that a torus is obtained. The so-obtained graph is regular.

\item \textit{Random Geometric Graph} with radius $r=0.3$: sensors are deployed in the square $[0,1]\times[0,1]$, their positions
being randomly generated with a uniform distribution; links are switched on between two sensors whenever the distance is less than $r$. 
 We only envisage connected realizations. 
\end{enumerate} 

From  Theorem \ref{teo:convergenza}, $P$ is required to possess positive eigenvalues: our intuition is that these
hypotheses, that are useful to prove the convergence of the IA, are not really necessary. We test this conjecture on the ring
graph, whose eigenvalues are known  \cite{Davis} to be $
\lambda_m=\frac{1}{3}\left(1+2\cos\left(\frac{2\pi m}{N}\right)\right),\quad m=0,\ldots,N-1$
and which are not necessarily positive.\\

\textbf{Algorithms}: We implement and compare the following algorithms: IA with  $\gamma^{(t)}= 1/t^{\zeta}$ for different choices of $\zeta\in\{0.5,0.7,0.9\}$, and the two centralized iterative algorithms IML, and EM described in Section \ref{iterative}.\\

\begin{figure}[h]
\begin{center}
\subfigure[Complete graph.]
{\label{fig1} 
\begin{minipage}[c]{0.47\columnwidth}
\centering
\includegraphics[width=1\columnwidth]{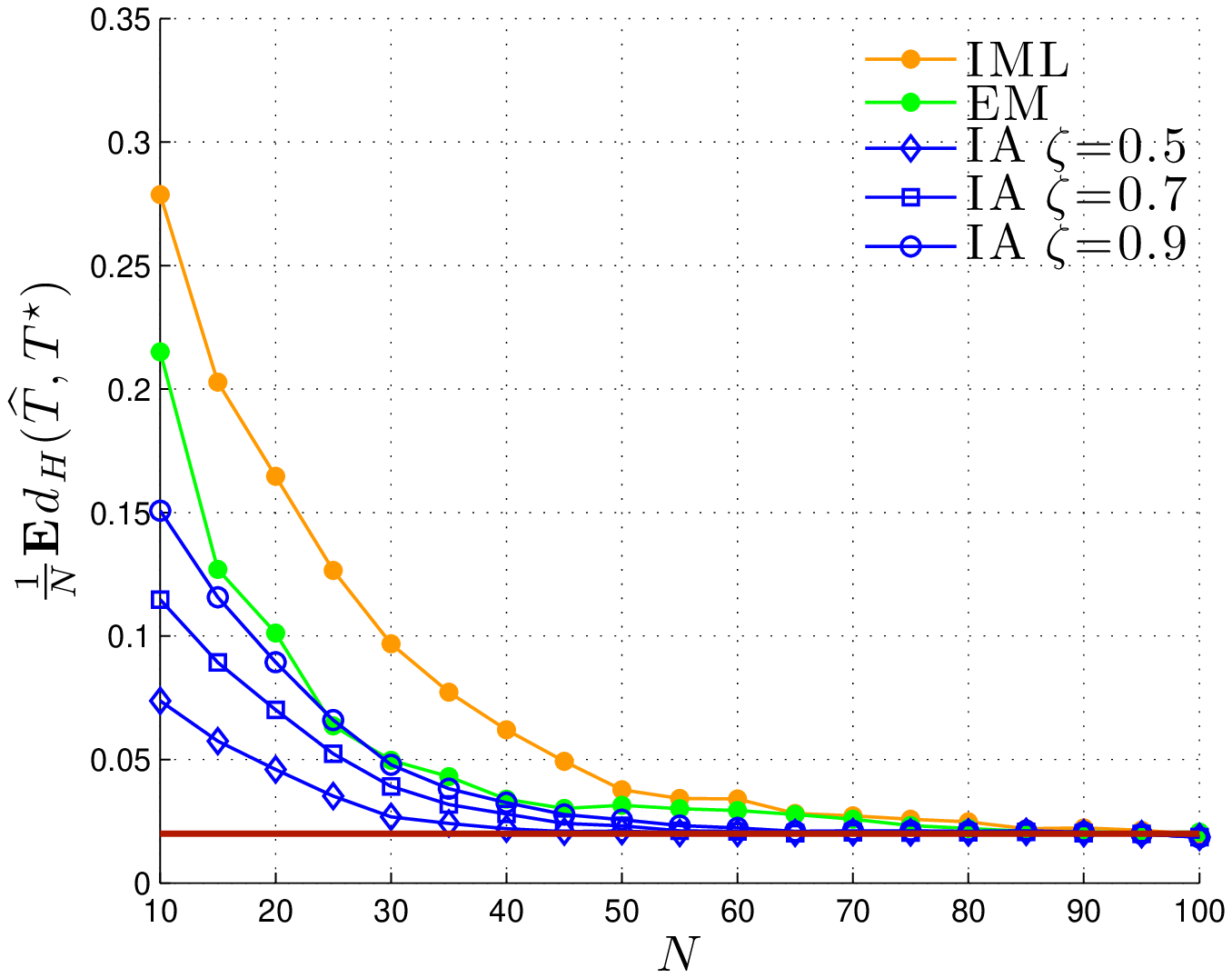}
\end{minipage}}
\subfigure[ Ring.]
{\label{fig2} 
\begin{minipage}[c]{0.47\columnwidth}
\centering
\includegraphics[width=1\columnwidth]{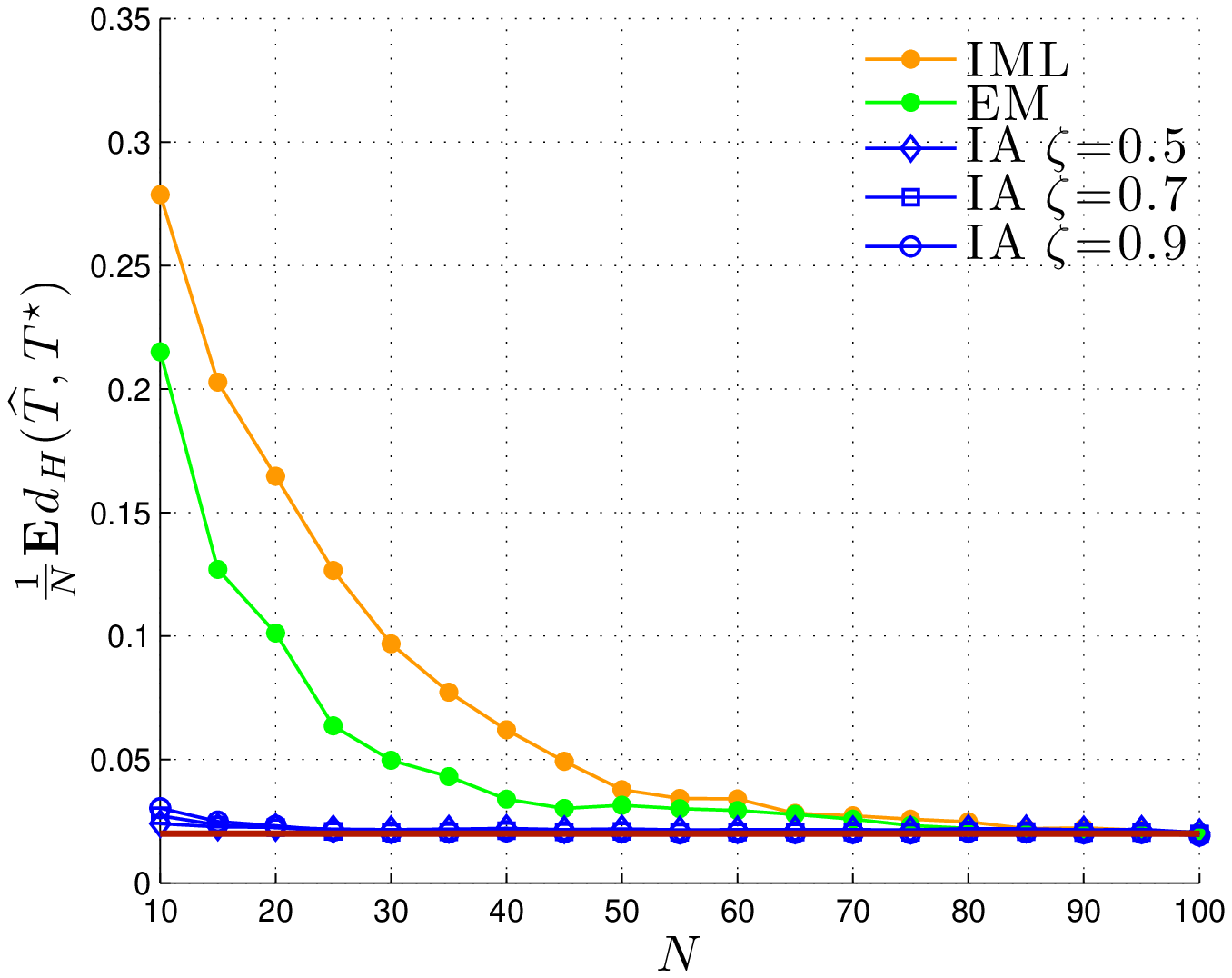}
\end{minipage}}
\end{center}
\begin{center}
\subfigure[Grid.]
{\label{fig3} 
\begin{minipage}[c]{0.47\columnwidth}
\centering
\includegraphics[width=1\columnwidth]{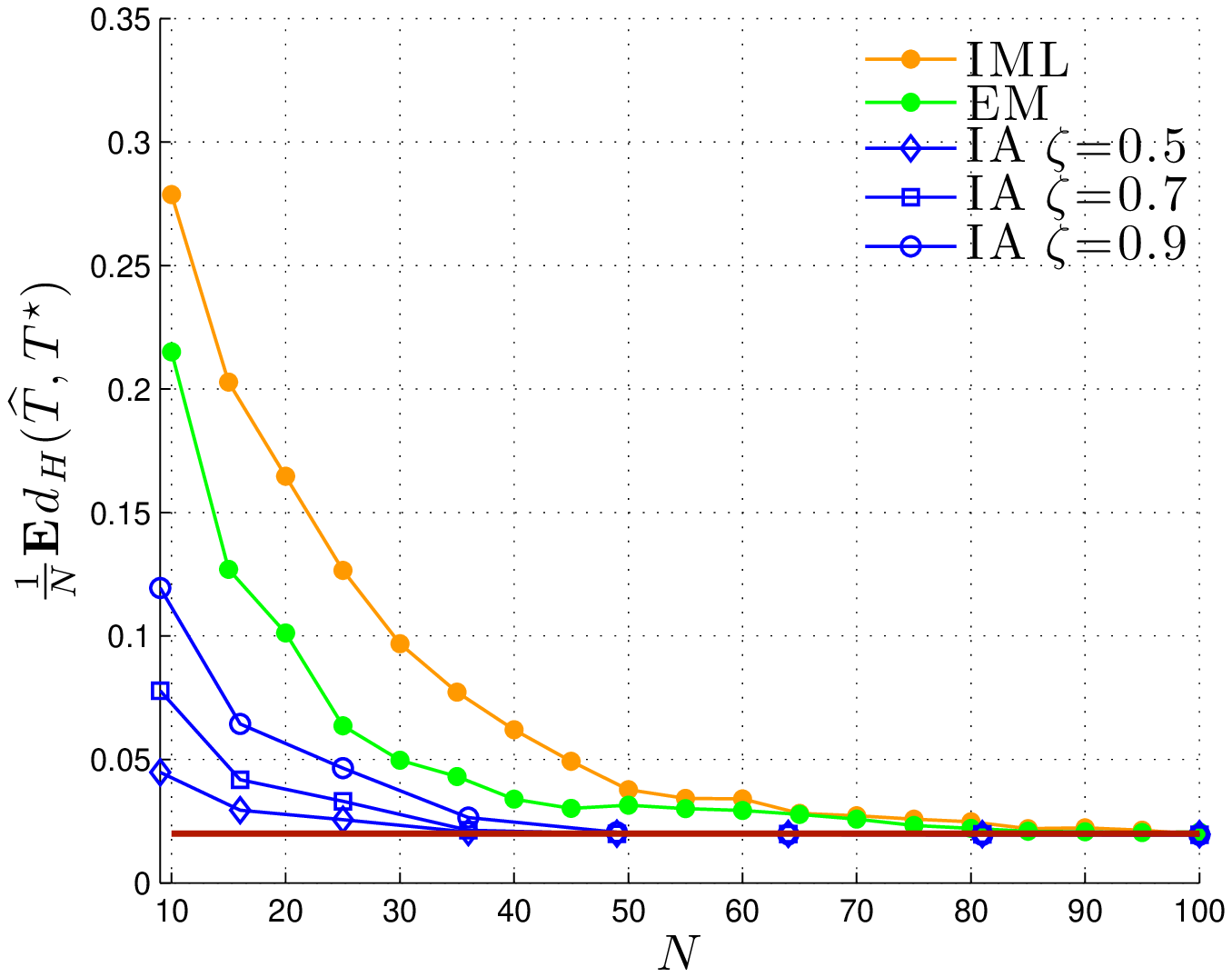}
\end{minipage}}
\subfigure[Random geometric graph.]
{\label{fig4} 
\begin{minipage}[c]{0.47\columnwidth}
\centering
\includegraphics[width=1\columnwidth]{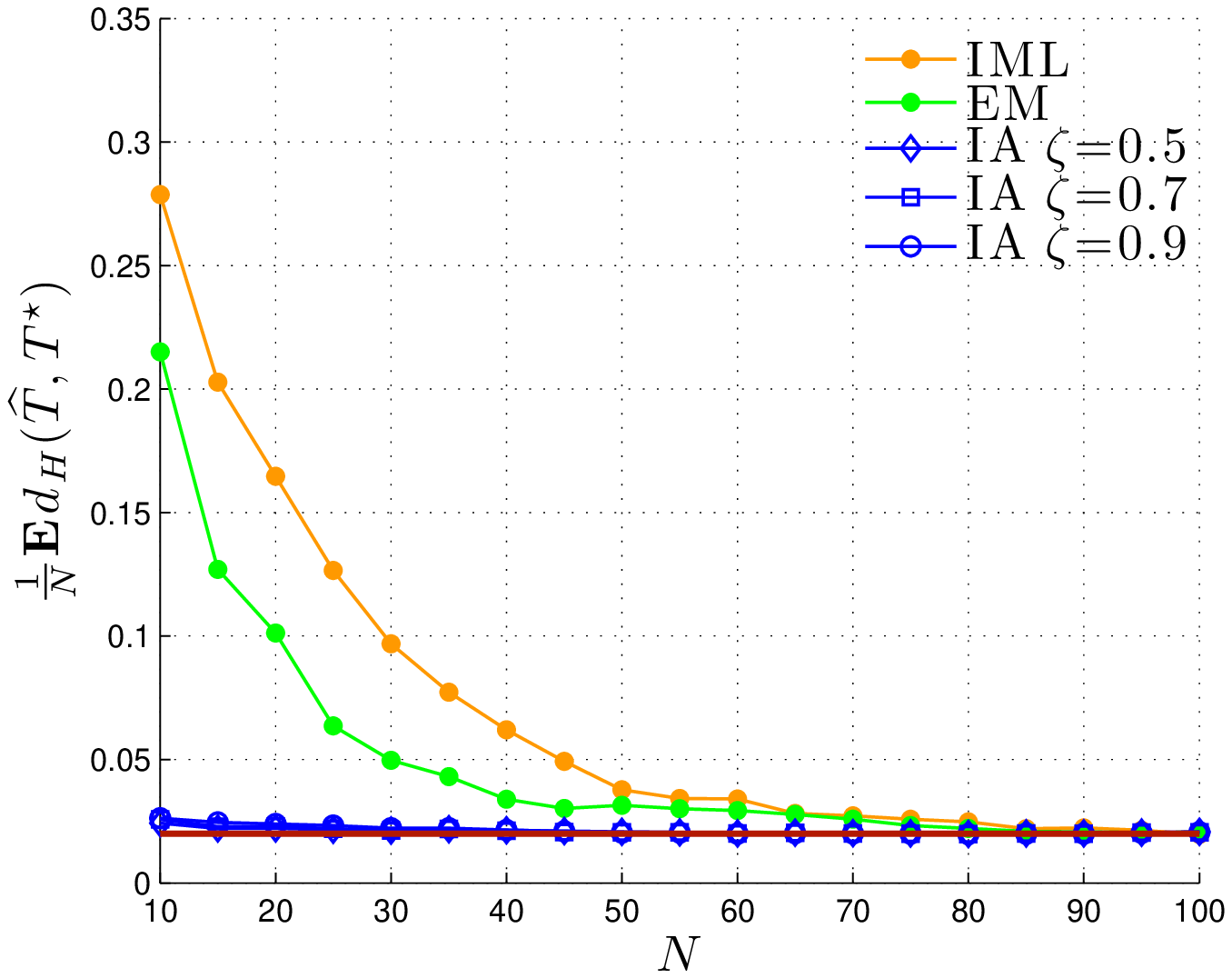}
\end{minipage}}
\caption{Relative classification error}\label{fig:fig1}
\end{center}
\end{figure}

\textbf{Outcomes}: we show the performance of the aforementioned algorithms in terms of classification error and of mean square error on
the global parameter, in function of the number of sensors $N$. All the outcomes are obtained averaging over $400$ Monte Carlo
runs.

We observe that the classification error (Figure \ref{fig:fig1}) converges for $N\to\infty$ for all the considered algorithms. On the other hand, when $N$ is small, IA performs better than IML and EM, no matter which graph topology has been chosen: this suggests that decentralization is then not a drawback for IA.
Moreover, for smaller $\gamma^{(t)}$ (i.e., slowing down the
procedure), we obtain better IA  performance in terms of classification. Nevertheless, this is not universally true: in other simulations, in fact, we have noticed that if 
 $\gamma^{(t)}$ is \textit{too} small, the performance are worse. This is not surprising, since $\gamma^{(t)}$ determines the weights assigned to the consensus and input driven parts, whose contributions must be somehow balanced in order to obtain the best solution. An important point that we will study in future is the optimization of $\gamma^{(t)}$, whose choice may in turn depend on the graph topology and on the weights assigned in the matrix $P$.

Analogous considerations can be done for the mean square error on $\theta$: when $N$ increases, the mean square error decays to zero.

We remark that convergence is numerically shown also for the ring topology, which is not envisaged by our theoretical analysis. Hence, our guess is  that convergence should be proved even under weaker hypotheses on matrix $P$.

For the interested reader, a graphical user interface of our algorithm is available and downloadable on http://calvino.polito.it/$\sim$fosson/software.html.


\section{Concluding remarks}\label{conclusion}
In this paper, we have presented a fully distributed algorithm
for the simultaneous estimation and classification in a sensor network, given from noisy measurements.
The algorithm only requires the local cooperation among
units in the network. Numerical simulations show remarkable
performance. The main contribution includes the convergence of the algorithm to a local maximum of the centralized ML estimator. The performance of the algorithm has been also studied when the network size is large, proving that the solution of the proposed algorithm concentrates around the classical ML solution.

Different variants are possible, for example the generalization to multiple classes with unknown prior probabilities should be inferred. The choice of sequence $\{\gamma^{(t)}\}_{t\in\mathbb{N}}$ is critical, since it influences both convergence time and final accuracy; the determination of a protocol for the adaptive search of sequence $\{\gamma^{(t)}\}_{t\in\mathbb{N}}$ is left for a future work.

\section{Acknowledgment}
The authors wish to thank Sandro Zampieri for bringing the problem to our attention and Luca Schenato for useful discussions. F. Fagnani and C. Ravazzi further acknowledge the financial support provided by MIUR under the PRIN project no. 20087W5P2K.

\bibliographystyle{ieeetr}
\bibliography{clus}

\appendix
\section{Proof of Theorem \ref{teo:convergenza}}\label{appA}
Consider the discrete-time dynamical system defined by the update equations \eqref{munu} and \eqref{sys}: 
the proof of its convergence is obtained through intermediate steps.

\begin{enumerate}
\item First, we show that, for sufficiently large $t $, vectors
$\mu^{(t)},\nu^{(t)}$, and $\widehat{\theta}^{(t)}$ are close to
consensus vectors and we prove their convergence, assuming
$\widehat{\omega}^{(t)}$ has already stabilized. \item Second, we prove
the stabilization of $\widehat{\omega}^{(t)}$ in finite time, by
modelling the system in \eqref{munu} and \eqref{sys} as a
switching dynamical system. \item Finally, combining these facts
together we conclude the proof.
\end{enumerate}

\subsection{Towards consensus}
We start with some notation: let
$\Omega:=I-N^{-1}\mathbbm{1}\mathbbm{1}^{\mathsf{T}}$; given
$x\in\mathbb{R}^{\mathcal V}$, let
$\overline{x}:=N^{-1}\mathbbm{1}^{\mathsf{T}}x$ so that $ x=
\overline{x}\mathbbm{1}+\Omega x$.

Given a bounded sequence $u^{(t)}\in\mathbb{R}^N$, consider the
dynamics
\begin{equation}\label{sistemadinamico}
  \begin{split}
 x^{(t+1)}=\left(1-\gamma^{(t)}\right) P x^{(t)}+\gamma^{(t)} u^{(t)}~~~t\in\mathbb{N}
  \end{split}
\end{equation}
where $x^{(0)}$ is any fixed vector, and where, we recall the
standing assumptions,
\begin{enumerate}
\item[(a)] $\gamma^{(t)}\in (0,1)$, $\gamma^{(t)}\geq 1/t$,
$\gamma^{(t)}\searrow 0$ and
$\gamma^{(t)}=\gamma^{(t+1)}+o(\gamma^{(t+1)})$ for $t\to
+\infty$; \item[(b)] $P\in \mathbb{R}_+^{\mathcal V\times \mathcal V}$ is a
stochastic, symmetric, primitive matrix with positive eigenvalues.
\end{enumerate}

A useful fact consequence of the assumptions on $\gamma^{(t)}$, is the following:
\begin{equation}\label{gammadis1}
\prod\limits_{s=t_0}^{t-1}(1-\gamma^{(s)})\leq e^{-\sum\limits_{s=t_0}^{t-1}1/s}\leq \frac{t_0}{t}\leq t_0\gamma^{(t)}\end{equation}
for any choice of $t\geq t_0>0$.

On the other hand, as a consequence of the assumptions of $P$ (see \cite{Perron}) we have that $P^t\to N^{-1}\mathbbm 1\mathbbm 1^T$, or equivalently that $P^t\Omega\to 0$ for $t\to +\infty$. More precisely,  we can order the eigenvalues of $P$ as  $1=\mu_1> \mu_2\geq\cdots\geq\mu_N\geq 0$, and we have that $||P^t\Omega||\leq \mu_2^t$.

\begin{lemma}\label{retta} It holds
$$\Omega x^{(t)}=O(\gamma^{(t)})\,,\quad {\rm for}\; t\to +\infty.$$
\end{lemma}

\begin{proof}
From \eqref{sistemadinamico} and the fact that $\Omega P=P\Omega$ we get, for any fixed $t_0$ and $t\geq t_0$,
\begin{equation}\label{soldyn}
\Omega
x^{(t+1)}=\prod_{s=t_0}^{t}\left(1-\gamma^{(s)}\right)P^{t}\Omega
x^{(t_0)} +\sum_{s=t_0}^{t}\prod_{k=s+1}^{t}
\left(1-\gamma^{(k)}\right)\gamma^{(s)}P^{t-s} \Omega u^{(s)}.
\end{equation}
This yields
\begin{align}\label{omegax}
||\Omega x^{(t+1)}||_2&\leq\prod_{s=t_0}^{t}\left(1-\gamma^{(s)}\right) ||\Omega x^{(t_0)}||_2+ \sum_{s=t_0}^{t} \prod_{k=s+1}^{t}
\left(1-\gamma^{(k)}\right)\gamma^{(s)}|\mu_2|^{t-s}||u^{(s)}||_2\nonumber\\
&\leq\prod_{s=t_0}^{t}\left(1-\gamma^{(s)}\right) || \Omega
x^{(t_0)} ||_2 + K\sum_{s=t_0}^{t} \prod_{k=s+1}^{t}
\left(1-\gamma^{(k)}\right)\gamma^{(s)}|\mu_2|^{t-s}
\end{align}
with $K:=\max_s||u^{(s)}||_2$.

Fix now $0<\varepsilon<1-|\mu_2|$ and let $t_0\in\mathbb{N}$ be
such that $\frac{\gamma^{(t+1)}}{\gamma^{(t)}}\in (1-\varepsilon,
1)$ for all $t\geq t_0$. Hence, for $t\geq s\geq t_0$, we have
that $\gamma^{(s)}<\frac{\gamma^{(t)}}{(1-\varepsilon)^{t-s}}$.
Consider now the estimation \eqref{omegax} with this choice of
$t_0$. We get

%
\begin{align*}
||\Omega
x^{(t+1)}||_2&\leq\prod_{s=t_0}^{t}\left(1-\gamma^{(s)}\right)
||\Omega x^{(t_0)}||_2
+ K \gamma^{(t)}\sum_{s=t_0}^{t}\left(\frac {|\mu_2|}{1-\varepsilon}\right)^{t-s}\\
&\leq\prod_{s=t_0}^{t}\left(1-\gamma^{(s)}\right) ||\Omega
x^{(t_0)}||_2+\frac
{K\gamma^{(t)}}{1-\frac{|\mu_2|}{1-\varepsilon}}.
\end{align*}
Using now (\ref{gammadis1}) the proof is completed.
\end{proof}

\begin{proposition}\label{input_costante}
 If $\exists\  t_0\in\mathbb{N}$ s.t. $u^{(t)}=u$ $\forall t\geq t_0$ then
$$
\lim_{t\rightarrow+\infty} x^{(t)}=\overline{u}\mathbbm{1}.
$$
\end{proposition}
\begin{proof}
Write $x^{(t)}=\overline{x}^{(t)}\mathbbm{1}+\Omega x^{(t)}$ and
notice that from Lemma \ref{retta} it is sufficient to prove that
$\lim_{t\rightarrow+\infty}\overline{x}^{(t)} \mathbbm{1} =
\overline{u} \mathbbm{1}.$ From (\ref{sistemadinamico}) and the fact that ${\mathbbm 1}^TP={\mathbbm 1}^T$, we obtain
$$
\overline{x}^{(t)}
-\overline{u}=\prod_{s=t_0}^{t-1}(1-\gamma^{(s)})(\overline{x}^{(s)}
-\overline{u})
$$
which goes to zero from the non-summability of $\gamma^{(t)}$.
\end{proof}

We now apply these results to the analysis of
$\widehat\theta^{(t)}$. We start with a representation result.
\begin{lemma}\label{rapportomedie} It holds, for $t\to +\infty$,
\begin{equation}\label{thetahat}
 \widehat{\theta}^{(t)}=\frac{\bar{\mu}^{(t)}}{\bar{\nu}^{(t)}}\mathbbm{1}+
\frac{1}{\bar{\nu}^{(t)}}\Omega\left(\mu^{(t)}-\frac{\bar{\mu}^{(t)}}{\bar{\nu}^{(t)}}\nu^{(t)}\right)+
 o\left(\gamma^{(t)}\right).
  \end{equation}
\end{lemma}
\begin{proof}
For any $i\in\mathcal{V} $,
\begin{align*}
\frac{\mu_i^{(t)}}{\nu_i^{(t)}}-\frac{\bar{\mu}^{(t)}}{\bar{\nu}^{(t)}}&=\frac{\mu_i^{(t)}}{\nu_i^{(t)}}-\frac{\bar{\mu}^{(t)}}{\bar{\nu}^{(t)}} + \frac{\mu_i^{(t)}}{\bar{\nu}^{(t)}}- \frac{\mu_i^{(t)}}{\bar{\nu}^{(t)}}\\
&=\frac{\mu_i^{(t)}-\bar{\mu}^{(t)}}{\bar{\nu}^{(t)}}+\mu_i^{(t)}\left(\frac{1}{\nu_i^{(t)}}- \frac{1}{\bar{\nu}^{(t)}}\right)\\
&=\frac{1}{\bar{\nu}^{(t)}}\left(\Omega\mu ^{(t)}\right)_i-\frac{\mu_i^{(t)}}{\nu_i^{(t)} \bar{\nu}^{(t)}}\left(\Omega\nu ^{(t)}\right)_i.
\end{align*}
It follows from Lemma \ref{retta} that $\mu^{(t)}=\bar{\mu}^{(t)}\mathbbm{1}+O(\gamma^{(t)})$ and $\nu^{(t)}=\bar{\nu}^{(t)}\mathbbm{1}+O(\gamma^{(t)})$  for
$t\to +\infty$. This and the fact that $\bar{\nu}^{(t)}$ is bounded away from $0$ (indeed $\bar{\nu}^{(t)}\geq\alpha^{-2}$ for all $t>0$), yields
\begin{align*}
\frac{\mu_i^{(t)}}{\nu_i^{(t)} \bar{\nu}^{(t)}}\left(\Omega\nu ^{(t)}\right)_i
&=\frac{\bar{\mu}^{(t)}}{ \bar{\nu}^{(t)}}\left[\frac{\left(\Omega\nu
^{(t)}\right)_i}{\bar{\nu}^{(t)}}\right]\left(1+O\left(\gamma^{(t)}\right)\right)
\end{align*}
from which thesis follows.
\end{proof}

We can now present our first convergence result.

\begin{corollary}\label{corol: retta} It holds, for $t\to +\infty$,
$$\bar{ \widehat{ \theta}}^{(t)}=\frac{\bar{\mu}^{(t)}}{\bar{\nu}^{(t)}}+o(\gamma^{(t)})\,,\quad
\Omega \widehat{ \theta}^{(t)}=O(\gamma^{(t)}).$$
\end{corollary}

\begin{proof} Both relations are obtained from (\ref{thetahat}). The first one is immediate. The second one follows
from Lemma \ref{retta} and the fact that $\bar{\nu}^{(t)}$
stays bounded away from $0$.
\end{proof}

Corollary \ref{corol: retta} says that the estimate
$\widehat{\theta}^{(t)} $ is close to a consensus for sufficiently
large $t$. Something more precise can be stated if we know that 
if $\widehat{\omega}^{(t)}$ stabilizes at
finite time as explained in the next result.

\begin{corollary}\label{convergenza_finale}
If $\exists\  t_0\in\mathbb{N}$ s.t. $\widehat{\omega}^{(t)}=\widehat{\omega}^{IA}$ $\forall t\geq t_0$ then
  $$
\lim_{t\rightarrow+\infty} \widehat{\theta}^{(t)}
=\widehat{\theta}(\widehat\omega^{IA})=\frac{\sum_{i\in\mathcal V}{y_i}{\left[\widehat{\omega}^{IA}_i\right]^{-2}}}{\sum_{i\in\mathcal V}{\left
[\widehat{\omega}^{IA}_i\right]^{-2}}} \mathbbm{1}.
$$
\end{corollary}
\begin{proof} Proposition
\ref{input_costante} guarantees that $\mu^{(t)}$ and $\nu^{(t)}$
converge to $\frac{1}{N}\sum_{i\in\mathcal{V}}y_i[
\widehat{\omega}^{IA}_i]^{-2}\mathbbm{1}$ and
$\frac{1}{N}\sum_{i\in\mathcal{V}}
[\widehat{\omega}^{IA}_i]^{-2}\mathbbm{1}$, respectively. This yields the thesis.
\end{proof}

\subsection{Stabilization of $\widehat{\omega}^{(t)}$}
We are going to prove that vector $ \widehat{\omega}^{(t)} $ almost surely stabilizes
in finite time: this, by virtue of previous considerations will
complete our proof. To prove this fact will take lots of effort
and will be achieved through several intermediate steps.

We start observing that, since $\widehat{\omega}^{(t)}$ can only assume
values in a finite set, equations in \eqref{munu} and \eqref{sys}
can be conveniently modeled by a switching system as shown below.

For reasons which will be clear below, in this subsection we will replace the configuration space $\{\alpha,\beta\}^{\mathcal V}$ with the augmented state space $\{\alpha,\beta+,\beta-\}^{\mathcal V}$. If $\omega\in\{\alpha,\beta+,\beta-\}^{\mathcal V}$, define
\begin{align*}
\Theta_{\omega}=\{x\in\mathbb{R}^{\mathcal V}: &|x_i-y_i|<\delta, \text{if }\omega_i=\alpha, x_i\geq y_i+\delta,\text{if }\omega_i=\beta+, x_i\leq y_i-\delta,\text{if }\omega_i=\beta-
\}.
\end{align*}

We clearly have $\mathbb{R}^{\mathcal V}=\bigcup_{\omega\in\{\alpha,\beta+,\beta-\}^{\mathcal V}}\Theta_{\omega} $.

On each $\Theta_\omega$ the dynamical system is linear. Indeed, define the maps
$f_{\omega}:\mathbb{R}\times \mathbb{R}^{\mathcal V}\rightarrow\mathbb{R}^{\mathcal V}$ and $g_{\omega}:\mathbb{R}\times\mathbb{R}^{\mathcal V}\rightarrow\mathbb{R}^{\mathcal V}$ by
\begin{align*}
[f_{\omega}(t,x)]_i&=(1-\gamma^{(t)})[Px^{(t)}]_i+\gamma^{(t)}\frac{y_i}{\omega_i^2}\\
[g_{\omega}(t,x)]_i&=(1-\gamma^{(t)})[Px^{(t)}]_i+\gamma^{(t)}\frac{1}{\omega_i^2}
\end{align*}
where, conventionally, $\omega_i^2=\beta^2$ if $\omega_i=\beta+, \beta-$.
Then, if $\widehat{\theta}^{(t)}\in\Theta_{\omega}$, \eqref{mu}, \eqref{nu}, and \eqref{sys1} can be written as
$$\mu^{(t+1)}=f_{{{\omega}}}(t,\mu^{(t)})\qquad\nu^{(t+1)}=g_{{{\omega}}}(t,\nu^{(t)})$$
$$\widehat{\theta}^{(t+1)}_i={\mu_i^{(t+1)}}/{\nu_i^{(t+1)}}.
$$
Notice that this is a closed-loop switching system, since the switching policy is determined by $\widehat{\theta}^{(t)}$. It is clear that the stabilization of $\widehat{\omega}^{(t)}$ is equivalent to the fact that there exist an ${\omega}\in\{\alpha,\beta+,\beta-\}^N$ and a time $\widetilde{t}$ such that $\widehat{\theta}^{(t)}\in\Theta_{{\omega}}$ for all $t\geq \widetilde{t}$.

From Corollary \ref{convergenza_finale} candidate limit points for $\widehat{\theta}^{(t)}$ are
$${\widehat\theta}(\omega)\mathbbm{1}=\frac{\sum_{i\in\mathcal{V}}y_i\omega_i^{-2}}{\sum_{i\in\mathcal{V}}\omega_i^{-2}}\mathbbm{1}\qquad \omega\in\{\alpha,\beta+,\beta-\}^N.$$

Also, from Proposition \ref{corol: retta},  the dynamics can be conveniently analysed by studying it in a neighborhood of the line
$\Lambda=\{\lambda\mathbbm{1}|\lambda\in\mathbb{R}\}$.

We now make an assumption which holds almost everywhere with
respect to the choice of $y_i$'s and, consequently, does not
entail any loss of generality in our proof.

\smallskip\noindent
{\bf ASSUMPTION:}
\begin{itemize}
\item $y_i-y_j\not\in\{0, \pm\delta, \pm 2\delta\}$ for all $i\neq
j$; \item ${\widehat\theta}(\omega)-y_i\not\in\{\pm\delta\}$ for all
$\omega\in \{\alpha,\beta+,\beta-\}^{\mathcal V}$ and for all $i$.
\end{itemize}
This assumption has a number of consequences which will be used later on:
\begin{enumerate}
\item[(C1)] ${\widehat\theta}(\omega){\mathbbm 1},
y_i{\mathbbm 1}\in\bigcup_{\omega\in\{\alpha,\beta+,\beta-\}^{\mathcal V}}{\rm
int}(\Theta_\omega)$ for all $\omega\in
\{\alpha,\beta+,\beta-\}^{\mathcal V}$ and for all $i\in\mathcal V$;

\item[(C2)]
$\Lambda\cap\bar\Theta_{\omega}\cap\bar\Theta_{\omega'}\cap\bar\Theta_{\omega''}=\emptyset$
for any triple of distinguished $\omega, \omega^\prime, \omega
''$. In other terms, $\Lambda$ always crosses boundaries among
regions $\Theta_\omega$ at internal point of faces.
\end{enumerate}

We now introduce some further notation, which will be useful in
the rest of the paper.
\begin{gather*}
\Theta^{\epsilon}:=\{x\in\mathbb{R}^{\mathcal V}:||\Omega x||_2<\epsilon\},\quad\Theta^{\epsilon}_{\omega}:=\Theta^{\epsilon}\cap \Theta_{\omega}\\
\Gamma:=\{\omega\in\{\alpha,\beta+,\beta-\}^{\mathcal V}: \ \Theta_{\omega}\cap\Lambda\neq\emptyset\}.\end{gather*}

For any $\omega\in \Gamma$ consider
$$
\Pi_{\omega}=\{\pi=\bar\Theta_{\omega}\cap\bar\Theta_{\omega'}:\ \mathrm{d_H}(\omega,\omega')=1,\ \pi\cap\Lambda=\emptyset\}
$$
and define $ \sigma_{\omega}:=\min_{\pi\in
\Pi_{\omega}}\mathrm{d}(\Theta_\omega\cap\Lambda,\pi)>0$\begin{footnote}{
$\mathrm{d}(\Theta_\omega\cap\Lambda,\pi)$ denotes the distance
between the two sets $\Theta_\omega\cap\Lambda$ and the set
$\pi$}\end{footnote}.

In the sequel, we will use the natural ordering on $\Lambda$: given the sets $X,Y\subseteq\Lambda$, $X<Y$ means that $x<y$ for all $x\in X$ and $y\in Y$.

\begin{definition}
Given two elements $\omega,\ \omega'\in\Gamma $, we say that $\omega'$ is the future-follower of $\omega$ (or also that $\omega$ is the past-follower of $\omega'$) if the following happens:
\begin{enumerate}\label{conditions_adjacent}
\item[(A)] There exists $i_0$ such that $
\omega_i=\omega'_{i}\text{ for all }\, i\neq i_0\text{ and
}\omega_{i_0}\neq\omega'_{i_0}$; \item[(B)] $
\Theta_{\omega}\cap\Lambda<\Theta_{\omega'}\cap\Lambda$.
\end{enumerate}
\end{definition}
Notice that, in order for $\omega$ and $\omega'$ to satisfy
definition above, it must necessarily happen that either
$\omega_{i_0}=\alpha$ and $\omega'_{i_0}=\beta+$, or
$\omega_{i_0}=\beta-$ and $\omega'_{i_0}=\alpha$. Given
$\omega\in\Gamma$, its future-follower (if it exists) will be
denoted by $\omega^+$. It is clear that (because of property (C2)
described above) that we can order elements in $\Gamma$ as
$\omega^1,\omega^2,\dots ,\omega^M$ in such a way that
$\omega^{r+1}=(\omega^r)^+$ for every $r=1,\dots ,M-1$.

Given $\omega\in\Gamma$, consider the following subsets of $\mathbb{R}^N$ (see Fig.
\ref{fig3b}):
\begin{align*}
\mathcal{M}_{\omega}^{\epsilon}&:=\left\{x\in\Theta_{\omega}^{\epsilon}:\overline{x}\mathbbm{1}+\Omega z\in
\Theta_{\omega}^{\epsilon},~\forall z: ||z||_2<\epsilon \right\}
\\
\mathcal{L}_{\omega,\omega^+}^{\epsilon}&:=\left\{x\in\Theta^{\epsilon}: \mathcal{M}_{\omega}^{\epsilon}\cap\Lambda<\bar
x<\mathcal{M}_{\omega^+}^{\epsilon}\cap\Lambda\right\}.
\end{align*}
(with the implicit assumption that
$\mathcal{L}_{\omega,\omega^+}^{\epsilon}=\emptyset$ if $\omega^+$
does not exist.) We clearly have
$\Theta^{\epsilon}=\bigcup_{\omega\in\Gamma}
\mathcal{M}_{\omega}^{\epsilon}
\cup\mathcal{L}_{\omega,\omega^+}^{\epsilon} $.

\begin{figure}[h!]\begin{center}
\includegraphics[height=0.4\columnwidth]{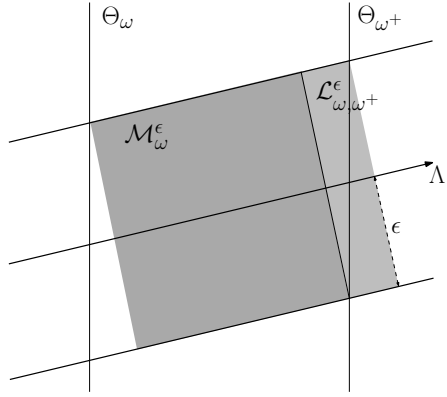}
\caption{Given the couple $(\omega,\omega')$ the sets $\mathcal{L}_{\omega,\omega'}^{\epsilon}$ and $\mathcal{M}_{\omega}^{\epsilon}$
are visualized.}
\label{fig3b}
\end{center}
\end{figure}

Notice that, because of property (C1), we can always choose
$\epsilon_0\in(0,\min_{\omega\in\Gamma} \sigma_{\omega})$ such
that
\begin{align*}&{\widehat\theta}(\omega)\mathbbm{1}, y_{i}\mathbbm{1}\in\bigcup_{\omega'\in\Gamma}\mathcal{M}_{\omega'}^{\epsilon_0}\qquad \forall
\omega\in\Gamma, \forall i\in\mathcal{V}.
\end{align*}
This implies that there exists $\tilde c>0$ such that
\begin{equation}\label{boundary}\mathrm{d}\left(\bigcup_{\omega'\in\Gamma}\partial_\Lambda\left(\mathcal{M}_{\omega'}^{\epsilon}
\cap\Lambda\right), \{{\widehat\theta}(\omega), y_i\}\right)\geq \tilde c,\quad\forall\epsilon\leq\epsilon_0
\end{equation} where $\partial_\Lambda(\cdot)$ denotes the boundary of a set in the relative topology of $\Lambda$.

Fix now $\epsilon\leq\epsilon_0$ and choose $t_\epsilon$ such that $\widehat{\theta}^{(t)}\in\Theta^{\epsilon}$ for all $t\geq
t_\epsilon$ (it exists by Corollary \ref{corol: retta}). From now on we consider times $t\geq t_\epsilon.$
Our aim is to prove through intermediate steps the following facts
\begin{itemize}
\item [F1)] if ${\widehat\theta}(\omega)\in\mathcal{M}^{\epsilon}_{\omega}$
then $\mathcal{M}^{\epsilon}_{\omega} $ is an \emph{asymptotically
invariant} set for $\widehat{\theta}^{(t)}$, namely, when $t$ is
sufficiently large, if $\widehat{\theta}^{(t)}\in
\mathcal{M}^{\epsilon}_{\omega}$ then $\widehat{\theta}^{(t+1)}\in
\mathcal{M}^{\epsilon}_{\omega}$; \item [F2)] if ${\widehat\theta}(\omega)\notin\mathcal{M}^{\epsilon}_{\omega}$ then
$\widehat{\theta}^{(t)}\notin \mathcal{M}^{\epsilon}_{\omega} $
for $t$ sufficiently large; \item[F3)]
$\widehat{\theta}^{(t)}\notin\bigcup_{\omega\in\{\alpha,\beta+,\beta-\}^{\mathcal V}}\mathcal{L}^{\epsilon}_{\omega,\omega^+}
$ for $t$ sufficiently large.
\end{itemize}

\subsubsection*{F1) Asymptotic invariance of $\mathcal{M}_{\omega}^{\epsilon}$ when ${\widehat\theta}(\omega)\mathbbm{1}\in\mathcal{M}^{\epsilon}_{\omega}$}
\text{\\}

\begin{lemma}\label{lemma: moto parallelo} If $\widehat{\theta}^{(t)}\in\Theta_{\omega}$ then there exists $c^{(t)}\in[\alpha^2/\beta^2, \beta^2/\alpha^2]$ and 
$r^{(t)}=o(\gamma^{(t)})$ for $t\to +\infty$
such that
\begin{equation}\label{parallel}\overline{\widehat\theta}^{(t+1)}=\overline{\widehat\theta}^{(t)}+
c^{(t)}\gamma^{(t)}\left({\widehat\theta}(\omega)-\overline{\widehat\theta}^{(t)}\right)+r^{(t)}
\end{equation}
\end{lemma}

\begin{proof}
If $\widehat{\theta}^{(t)}\in\Theta_{\omega} $ then
\begin{align*}
\frac{\overline{\mu}^{(t+1)}}{\overline{\nu}^{(t+1)}}-\frac{\overline{\mu}^{(t)}}{\overline{\nu}^{(t)}}
&=\frac{(1-\gamma^{(t)}) \overline{\mu}^{(t)}+\gamma^{(t)}N^{-1}
\sum_{i=1}^{N}y_i\omega_i^{-2}}{(1-\gamma^{(t)}) \overline{\mu}^{(t)}+\gamma^{(t)}N^{-1}\sum_{i=1}^{N}
\omega_i^{-2}}-\frac{\overline{\mu}^{(t)}}{\overline{\nu}^{(t)}} \\
&= \frac{\overline{\nu}^{(t)}\gamma^{(t)}N^{-1}\sum_{i=1}^{N}y_i\omega_i^{-2}-\overline{\mu}^{(t)}\gamma^{(t)}N^{-1}\sum_{i=1}^{N}\omega_i^{-2}}{\overline{\nu}^{(t+1)} \overline{\nu}^{(t)}}\\
&=\gamma^{(t)}\frac{
N^{-1}\sum_{i=1}^{N}\omega_i^{-2}}{\overline{\nu}^{(t+1)}
}{\left({\widehat\theta}(\omega)-\frac{\overline{\mu}^{(t)}}{\overline{\nu}^{(t)}}\right)}
\end{align*}

Choosing
$c^{(t)}= \frac{ N^{-1}\sum_{i=1}^{N}\omega_i^{-2}}{\overline{\nu}^{(t+1)} }\in[\alpha^{2}/\beta^2,\beta^{2}/\alpha^2]
$ and using Corollary \ref{corol: retta} thesis easily follows.
\end{proof}

\begin{proposition}[Proof of F1)]\label{inv}
There exists
$t'\geq t_\epsilon$ such that, if ${\widehat\theta}(\omega)\mathbbm{1}\in \Theta_{\omega}$, then
$$\widehat{\theta}^{(t)}\in \mathcal{M}^{\epsilon}_{\omega}\;\Rightarrow\;\widehat{\theta}^{(t+1)}\in
\mathcal{M}^{\epsilon}_{\omega}\quad\forall t\geq t'\,.$$
\end{proposition}
\begin{proof}
Consider the relation (\ref{parallel}).
If $\widehat{\theta}^{(t)}\in
\mathcal{M}^{\epsilon}_{\omega}$ and if $t$ is large enough so that  $c^{(t)}\gamma^{(t)}<1$ , we
have, by convexity, that
{\small{$$z:=\overline{\widehat\theta}^{(t)}+
c^{(t)}\gamma^{(t)}\left({\widehat\theta}(\omega)-\overline{\widehat\theta}^{(t)}\right)\in
\mathcal{M}^{\epsilon}_{\omega}.$$}} Moreover, because of
(\ref{boundary}) and the fact that $c^{(t)}$ is bounded away from
$0$, there exists $c'>0$ such that $\mathrm{d}(z, \partial
(\mathcal{M}^{\epsilon}_{\omega}\cap\Lambda))\geq c'\gamma^{(t)}$.
Proof is then completed by selecting $t'\geq t_\epsilon$ such that
$c^{(t)}\gamma^{(t)}<1$ and $|r(t)|<c'\gamma^{(t)}/2$ for all
$t\geq t'$.
\end{proof}

\subsection*{F2) Transitivity of $\mathcal{M}^{\epsilon}_{\omega} $ when ${\widehat\theta}(\omega)\mathbbm{1}\notin\mathcal{M}_{\omega}^{\epsilon}$}

Our next goal is to prove that if  ${\widehat\theta}(\omega)\mathbbm{1}\notin \mathcal{M}^{\epsilon}_{\omega} $, then, at a certain time
$t$, $\widehat\theta^{(t)}$ will definitively be outside $\mathcal{M}^{\epsilon}_{\omega}$. A technical lemma based on convexity arguments
is required.

\begin{lemma}\label{lemma:y_i0} Let $\omega\in\Gamma$ be such that there exists its future-follower $\omega^+$. Then,
%
$$\begin{array}{lcl}{\widehat\theta}(\omega)\mathbbm{1}>\Theta_{\omega}\cap\Lambda\;&\Rightarrow\; &{\widehat\theta}(\omega^+)\mathbbm{1}>\Theta_{\omega}\cap\Lambda\\
{\widehat\theta}(\omega^+)\mathbbm{1}
<\Theta_{\omega^+}\cap\Lambda\;&\Rightarrow\; &{\widehat\theta}(\omega)
\mathbbm{1} <\Theta_{\omega^+}\cap\Lambda.\end{array}$$
\end{lemma}

\begin{proof}
Suppose  $\omega_i=\omega^+_i,\forall i\neq i_0$ and
$\omega_{i_0}=\beta-$, $\omega^+_{i_0}=\alpha$ (the other case can
be treated in an analogous way). Pick $x'\in \Theta_{\omega}
\cap\Lambda$ and $x''\in\Theta_{\omega^+} \cap\Lambda$. From
$|x''-y_{i_0}|<\delta$, and $|x'-y_{i_0}|>\delta$ it immediately
follows that $ x''>y_{i_0}-\delta,\   x'<y_{i_0}-\delta $ and, in
particular, the fact
\begin{equation}\label{yi0}y_{i_0}\mathbbm{1}>\Theta_{\omega}\cap\Lambda\,.\end{equation}


%
Notice now that
\begin{align*}
{\widehat\theta}(\omega^+)
&= \frac{y_{i_0}\left(\frac{1}{\alpha^2}
-\frac{1}{\beta^2}\right)}{\sum\limits_{i\in\mathcal{V}}\frac{1}{{\omega^+_i}^2}} + \frac{\sum\limits_{i\in\mathcal{V}\setminus{i_0}}\frac{y_i}{{\omega^+_i}^2}+ \frac{y_{i_0}}
{\beta^2}}{\sum\limits_{i\in\mathcal{V}}\frac{1}{{\omega^+_i}^2}}\\
&= \frac{y_{i_0}\left(\frac{1}{\alpha^2}-\frac{1}{\beta^2}\right)}{\sum\limits_{i\in\mathcal{V}}\frac{1}{{\omega^+_i}^2}} +
\frac{{\widehat\theta}(\omega) \sum\limits_{i\in\mathcal{V}}\frac{1}{{\omega_i}^2}}{\sum\limits_{i\in\mathcal{V}}\frac{1}{{\omega^+_i}^2}}\\
&=\frac{y_{i_0}\left(\frac{1}{\alpha^2}-\frac{1}{\beta^2}\right)}{\sum\limits_{i\in\mathcal{V}}\frac{1}{{\omega^+_i}^2}}
+ \frac{{\widehat\theta}(\omega)
\left[\sum\limits_{i\in\mathcal{V}}\frac{1}{{\omega^+_i}^2}-\left(\frac{1}{\alpha^2}-\frac{1}{\beta^2}\right)\right]}{\sum\limits_{i\in\mathcal{V}}\frac{1}{{\omega^+_i}^2}}.
\end{align*}
%
In Figures \ref{fig: regioni} and \ref{fig: regioni3} a picture of the various points is depicted when ${\widehat\theta}(\omega)>\Theta_{\omega}\cap\Lambda$.

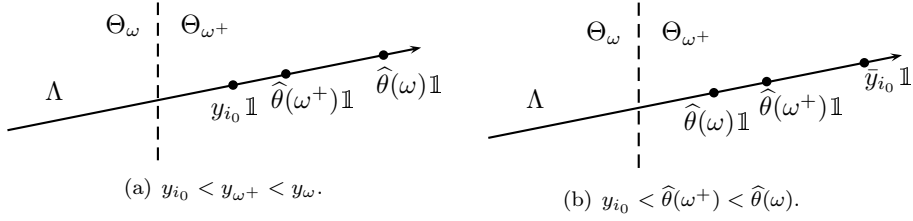
\begin{figure}[h]
\begin{center}
\subfigure[$y_{i_0}<y_{\omega^+}< y_{\omega}$.]
{\label{subfig: regioni1} 
\begin{minipage}[c]{0.45\columnwidth}
\centering

\begin{pspicture}(2,0.1)(8,2.5)
\psline{->}(2,0.7)(7.5,1.8)
\psline[linestyle=dashed]{-}(4,0.3)(4,2.4)
\psdots(5, 1.3)
\psdots(5.7, 1.45)
\psdots(7, 1.7)
\put(2.5,1.1){$ \Lambda$}\put(3.3,2){$\Theta_{\omega}$}\put(4.3,2){$\Theta_{\omega^+}$} \put(4.7,0.9){$y_{i_0}\mathbbm{1}$} \put(5.5,1){${\widehat\theta}(\omega^+)\mathbbm{1}$}\put(6.9,1.2){${\widehat\theta}(\omega)\mathbbm{1}$}
\end{pspicture}
\end{minipage}}
\subfigure[ $y_{i_0}<{\widehat\theta}(\omega^+)<{\widehat\theta}(\omega)$.]
{\label{subfig: regioni2} 
\begin{minipage}[c]{0.45\columnwidth}
\centering

\begin{pspicture}(1.7,0.1)(8,2.7)
\psline{->}(2,0.7)(7.5,1.8)
\psline[linestyle=dashed]{-}(4,0.3)(4,2.5)
\psdots(5, 1.3)
\psdots(5.7, 1.45)
\psdots(7, 1.7)
\put(2.5,1.1){$ \Lambda$}\put(3.3,2){$\Theta_{\omega}$}\put(4.3,2){$\Theta_{\omega^+}$} \put(4.6,0.8){${\widehat\theta}(\omega)\mathbbm{1}$} \put(5.6,1){${\widehat\theta}(\omega^+)\mathbbm{1}$}\put(7,1.4){$\bar y_{i_0}\mathbbm{1}$}
\end{pspicture}
\end{minipage}}

\caption{$\omega_{i_0}=\beta,{\widehat\theta}(\omega)\mathbbm{1} >\Theta_{\omega}\cap\Lambda$}\label{fig: regioni}
\end{center}
\end{figure}

\begin{figure}[h]\centering\begin{pspicture}(1,0)(8,2)
\psline{->}(2,0.7)(7.5,1.8)
\psline[linestyle=dashed]{-}(4,0.3)(4,2.3)
\psdots(5, 1.3)
\psdots(5.7, 1.45)
\psdots(3, 0.9)
\put(2.5,1.1){$ \Lambda$}\put(3.3,1.9){$\Theta_{\omega}$}\put(4.3,1.9){$\Theta_{\omega^+}$} \put(4.6,0.8){${\widehat\theta}(\omega)\mathbbm{1}$} \put(5.6,1){${\widehat\theta}(\omega^+)\mathbbm{1}$}\put(2.9,0.55){$\bar y_{i_0}\mathbbm{1}$}
\end{pspicture}
\caption{${\widehat\theta}(\omega)\mathbbm{1} >\Theta_{\omega}\cap\Lambda$}\label{fig: regioni3}
\end{figure}

A convexity argument and the use of (\ref{yi0}) now allow to
conclude.
\end{proof}

\begin{proposition}[Proof of F2)]\label{inv2}
If ${\widehat\theta}(\omega)\mathbbm{1}\notin\Theta_{\omega}$, then
there exists $t''$ such that $\widehat\theta^{(t)}\notin
\Theta_{\omega}^{\epsilon}$ $\forall t>t''$.
\end{proposition}

\begin{proof}
Suppose ${\widehat\theta}(\omega)\mathbbm{1}>\Theta_{\omega}\cap\Lambda$ (the case when is $<$ can be treated analogously).
Lemma \ref{lemma:y_i0} implies that
${\widehat\theta}(\omega^+)\mathbbm{1}>\Theta_{\omega}\cap\Lambda$.
Let $\tilde c$ be the constant given in (\ref{boundary}) and put $$A:=\{x\in \Theta^{\epsilon}_{\omega}\cup\Theta^{\epsilon}_{\omega^+}\;|\; \overline{x}\leq \alpha:=\min\{{\widehat\theta}(\omega),
{\widehat\theta}(\omega^+)\}-\tilde{c}/2\}.$$

Consider the relation (\ref{parallel}) and choose
$t_1$ in such a way that \begin{equation}\label{variation}\bar{\widehat{\theta}}^{(t+1)}-\bar{\widehat{\theta}}^{(t)}\leq c_2(\max\{y_i\}-\min\{y_i\})\gamma^{(t)}+r(t)<\tilde{c}/2\end{equation} and $|r(t)|<\alpha^2\tilde{c} \gamma^{(t)}/4\beta^2$ for
all $t\geq t_1$.
It also follows from (\ref{parallel}) that, if for some $t\geq t_1$
$\widehat{\theta}^{(t)}\in A$, then,
\begin{equation}\label{increase}\overline{\widehat\theta}^{(t+1)}\geq \overline{\widehat\theta}^{(t)}+ \alpha^2{\tilde{c}}\
\gamma^{(t)}/{4\beta^2}.\end{equation} Owing to the non-summability
of $\gamma^{(t)}$ it follows that if $\widehat{\theta}^{(t)}$
enters in $\Theta^{\epsilon}_{\omega}$ for some $t\geq t_1$, then,
in finite time it will enter into $A\setminus
\Theta^{\epsilon}_{\omega}$ and then it will finally exit $A$. In
particular there must exist $t_2\geq t_1$ such that
$\overline{\widehat{\theta}}^{(t_2)}> \alpha$. We now prove that
$\overline{\widehat{\theta}}^{(t_2)}{\mathbbm 1}>\Theta_\omega$ for every
$t\geq t_2$. If not there must exist a first time index $t_3>t_2$
such that $\overline{\widehat{\theta}}^{(t_3)}<\alpha-\tilde c$.
Because of (\ref{variation}), it must be that
$\overline{\widehat{\theta}}^{(t_3-1)}<\alpha-\tilde c/2$ but this
contradicts the fact that on $A$,
$\overline{\widehat{\theta}}^{(t)}$ is increasing (\ref{increase}).
\end{proof}

\subsection*{F3) Transitivity of $\bigcup_{\omega,\omega^+\in\{\alpha,\beta+\beta-\}^N}\mathcal{L}_{\omega,\omega^+}$}

We start with the following technical result concerning the general system (\ref{sistemadinamico}).
\begin{lemma}\label{quellodellenote}
Let $x^{(t)}$ be the sequence defined in \eqref{sistemadinamico}
and suppose that there exists a strictly increasing sequence of
switching times $\{\tau_k\}_{k=0}^{+\infty}$ such that
$$u_i^{(s+1)}=u_i^{(s)}\quad  \forall i\neq i_0\quad \text{and} \quad\forall s\in[\tau_0,+\infty[
$$
$$u_{i_0}^{(s)}= \begin{cases}
v'&\forall s\in I':=\bigcup_{k=0}^{+\infty} [\tau_{2k},\tau_{2k+1})\\
v''& \forall s\in I'':=\bigcup_{k=0}^{+\infty}
[\tau_{2k+1},\tau_{2k+2}).
\end{cases}
$$
Then, for every $\delta>0$, there exists $\bar t_\delta$ and two
sequences $a_\delta^{(t)}\geq 0$ and $b_\delta^{(t)}\leq\delta\gamma^{(t)}$, such that
\begin{align*}
 \left(\Omega\left(x^{(t+1)}-x^{(t)}\right)\right)_{i_0}\!\!\!=A_{\delta}^{(t)}\gamma^{(t)}
 \left(v'-v''\right) + b_\delta^{(t)}
\end{align*}
for $t\in I'$ with $t\geq \bar t_\delta$.
\end{lemma}
\begin{proof}
Let $\phi_j\in\R^{\mathcal V}$ be an orthonormal basis of  eigenvectors for $P$ relative to the eigenvalues
$1=\lambda_1>\lambda_2\geq\cdots\geq \lambda_N\geq 0$. Also assume we have chosen $\phi_1=N^{-1/2}\mathbbm{1}$.

We put
$$F^{(t)}:=\frac{\prod_{k=0}^{t}
\left(1-\gamma^{(k)}\right)}{\gamma^{(t)}}$$
and we notice that
$$\frac{F^{(s+1)}}{F^{(s)}}=(1-\gamma^{(s+1)})\frac{\gamma^{(s)}}{\gamma^{(s+1)}}\to 1\,,\;{\rm for}\; s\to +\infty.$$
Fix $\epsilon$ in such a way that $\lambda_2(1+\epsilon)<1$ and choose $s_0$ such that
$$\frac{F^{(s+1)}}{F^{(s)}}\leq 1+\epsilon\,,\;\forall s\geq s_0.$$

Let $t_0\geq s_0$ to be fixed later. From (\ref{soldyn}) we can
write

\begin{align}
 &\Omega(x^{(t+1)}-x^{(t)})=\nonumber\\
 &=\prod_{s=t_0}^{t-1}\left(1-\gamma^{(s)}\right)\left[\left(1-\gamma^{(t)}\right)P-I\right]P^{t-t_0}\Omega x^{(t_0)}\nonumber\\
& +\sum_{s=t_0}^{t}\prod_{k=s+1}^{t}
\left(1-\gamma^{(k)}\right)\gamma^{(s)}P^{t-s}\Omega u^{(s)}-
\sum_{s=t_0}^{t-1}\prod_{k=s+1}^{t-1}
\left(1-\gamma^{(k)}\right)\gamma^{(s)}P^{t-s-1}\Omega
u^{(s) }v\nonumber\\
&
=\prod_{s=t_0}^{t-1}\left(1-\gamma^{(s)}\right)\left[\left(1-\gamma^{(t)}\right)P-I\right]P^{t-t_0}\Omega
x^{(t_0)}\nonumber\\
&
+\gamma^{(t)}\sum_{s=t_0-1}^{t-1}P^{t-s-1}\frac{F^{(t)}}{F^{(s+1)}}\Omega
u^{(s+1)}-\gamma^{(t-1)}
\sum_{s=t_0}^{t-1}P^{t-s-1}\frac{F^{(t-1)}}{F^{(s)}}\Omega u^{(s)}\nonumber\\
&\label{repr1}
=\prod_{s=t_0}^{t-1}\left(1-\gamma^{(s)}\right)\left[\left(1-\gamma^{(t)}\right)P-I\right]P^t\Omega
x^{(t_0)}\\
&\label{repr2}+(\gamma^{(t)}-\gamma^{(t-1)})\sum_{s=t_0}^{t-1}P^{t-s-1}\frac{F^{(t-1)}}{F^{(s)}}\Omega u^{(s)}
+\gamma^{(t)}P^{t-t_0}\frac{F^{(t-1)}}{F^{(t_0)}}\Omega
u^{(t_0)}\\
&\label{repr3}+\gamma^{(t)}\sum_{s=t_0}^{t-1}P^{t-s-1}\left(\frac{F^{(t)}}{F^{(s+1)}}-\frac{F^{(t-1)}}{F^{(s)}}\right)\Omega
u^{(s+1)}\\
&\label{repr4}+
\gamma^{(t)}\sum_{s=t_0}^{t-1}P^{t-s-1}\frac{F^{(t-1)}}{F^{(s)}}\Omega\left( u^{(s+1)}-u^{(s)}\right).
\end{align}

It follows from the assumptions on $P$, the assumptions on $\gamma^{(t)}$ and relation (\ref{gammadis1}) that the terms
(\ref{repr1}) and (\ref{repr2}) are both $o(\gamma^{(t)})$ for $t\to +\infty$.
We now estimate (\ref{repr3}):
\begin{align}
&\left|\left|\sum_{s=t_0}^{t-1}P^{t-s-1}
\left(\frac{F^{(t)}}{F^{(s+1)}}-\frac{F^{(t-1)}}{F^{(s)}}\right)\Omega u^{(s+1)}\right|\right|_2=\nonumber\\
&\left|\left|\sum_{s=t_0}^{t-1}\frac{F^{(t-1)}}{F^{(s)}}\left(\frac{F^{(t)}}{F^{(t-1)}}\frac{F^{(s)}}{F^{(s+1)}}-1\right)P^{t-s-1}\Omega u^{(s+1)}\right|\right|_2\leq\nonumber\\
&\sum_{s=t_0}^{t-1}[\lambda_2(1+\epsilon)]^{t-s-1}\left|\left(\frac{F^{(s)}}{F^{(s+1)}}-1\right)\right|K
\leq\nonumber\\
&\label{estim2}\frac{K}{1-\lambda_2(1+\epsilon)}\beta_{t_0}
\end{align}
where $$K=\max||u^{(s)}||_2\,,\quad
\beta_{t_0}:=\sup\limits_{t\geq s\geq
t_0}\left|\left(\frac{F^{(t)}}{F^{(t-1)}}\frac{F^{(s)}}{F^{(s+1)}}-1\right)\right|.$$

We now concentrate on the component $i_0$ of the term
(\ref{repr4}). Using the spectral decomposition of $P$ and the
assumptions on $u^{(t)}$, we can write,
\begin{align}
&\left[\sum_{s=t_0}^{t-1}P^{t-s-1}\frac{F^{(t-1)}}{F^{(s)}}\Omega\left(u^{(s+1)}-u^{(s)}\right)\right]_{i_0}=\\
&\sum\limits_{j\geq 2}(\phi_j)_{i_0}^2\sum\limits_{h: t_0\leq\tau_h\leq
t-1}\lambda_j^{t-\tau_h}\frac{F(t-1)}{F(\tau_h-1)}(-1)^h(v'-v'').
\end{align}
If $t\in I'$, the above expression can be rewritten as
$$\sum\limits_{j\geq 2}(\phi_j)_{i_0}^2\sum\limits_{k: t_0\leq\tau_{2k}\leq t-1}\left[\lambda_j^{t-\tau_{2k}}\frac{F(t-1)}{F(\tau_{2k}-1)}-
\lambda_j^{t-\tau_{2k-1}}\frac{F(t-1)}{F(\tau_{2k-1}-1)}\right](v'-v'').$$
Notice that
$$\lambda_j^{t-\tau_{2k}}\frac{F(t-1)}{F(\tau_{2k}-1)}-
\lambda_j^{t-\tau_{2k-1}}\frac{F(t-1)}{F(\tau_{2k-1}-1)}=
\lambda_j^{t-\tau_{2k}}\frac{F(t-1)}{F(\tau_{2k}-1)}
\left(1-\lambda_j^{\tau_{2k}-\tau_{2k-1}}\frac{F(\tau_{2k}-1)}{F(\tau_{2k-1}-1)}\right)>0$$
(we have used the fact that $0\leq \lambda_j(1+\epsilon)<1$ for all $j\geq 2$). To complete the proof now proceed as follows. For a fixed $\delta
>0$, choose $t_0\geq s_0$ in such a way that (\ref{estim2}) is
below $\delta/2$. Then, fix $\bar t_\delta\geq t_0$ in such a way
that the summation of (\ref{repr1}) and (\ref{repr2}) is below
$\delta\gamma^{(t)}/2$ for $t\geq \bar t_\delta$. It is now sufficient  to define
$$a_{\delta}^{(t)}:=\sum\limits_j(\phi_j)_{i_0}^2\sum\limits_{k: t_0\leq\tau_{2k}\leq t-1}\left[\lambda_j^{t-\tau_{2k}-1}\frac{F(t-1)}{F(\tau_{2k}-1)}-
\lambda_j^{t-\tau_{2k-1}}\frac{F(t-1)}{F(\tau_{2k-1}-1)}\right]$$
and $b_{\delta}^{(t)}$ equal to the sum of the terms (\ref{repr1}), (\ref{repr2}), and (\ref{repr3}),
\end{proof}

\begin{proposition}[Proof of F3)]\label{limbo1}
There exists $t'''\in\mathbb{N}$ such that
\begin{equation}\label{trans}\widehat{\theta}^{(t)}\not\in\bigcup_{\omega,\omega'\in\{\alpha,\beta\}^N}\mathcal{L}_{\omega,\omega^+}\end{equation}
for all $t>t'''$.
 \end{proposition}
\begin{proof} In view of the results in Propositions \ref{inv} and
\ref{inv2}, and the fact that
$\widehat\theta^{(t+1)}-\widehat\theta^{(t)}$ goes to $0$ for
$t\to +\infty$, if (\ref{trans}) negation of (\ref{trans}) yields that there exists
$\omega\in\Gamma$ such that
$\widehat{\theta}^{(t)}\in\mathcal{L}_{\omega,\omega^+}$ for $t$
large enough. Now, if
$\widehat{\theta}^{(t)}\in\mathcal{L}_{\omega,\omega^+}\cap\Theta_{\omega}$
(or if
$\widehat{\theta}^{(t)}\in\mathcal{L}_{\omega,\omega^+}\cap\Theta_{\omega^+}$)
for $t$ sufficiently large, a straightforward application of
(\ref{parallel}) would imply that $\widehat{\theta}^{(t)}$ would
necessarily exit $\mathcal{L}_{\omega,\omega^+}$ in finite time.
Therefore, it must hold that $\widehat{\theta}^{(t)}$ keeps
switching, for large $t$, between
$\mathcal{L}_{\omega,\omega^+}\cap\Theta_{\omega}$ and
$\mathcal{L}_{\omega,\omega^+}\cap\Theta_{\omega^+}$.

From Lemma \ref{rapportomedie} and Corollary \ref{corol: retta} we can write
\begin{align*}
&\widehat{\theta}^{(t+1)}-\widehat{\theta}^{(t)}=\\
&\left(\frac{{\bar{\widehat{\theta}}}^{(t+1)}}{{\bar\nu}^{(t+1)}}-\frac{{\bar{\widehat{\theta}}}^{(t)}}{{\bar\nu}^{(t)}}\right){\mathbbm 1}+\frac{1}{{\bar\nu}^{(t)}}
\left[\Omega\left(\mu^{(t+1)}-\mu^{(t)}\right)-
\frac{{\bar\mu}^{(t)}}{{\bar\nu}^{(t)}}\Omega\left(\nu^{(t+1)}-\nu^{(t)}\right)\right]+o(\gamma^{(t)})
\end{align*}
Define now
$$I':=\{t\,|\, \widehat\theta^{(t)}\in\Theta_{\omega}\}\,,\quad I'':=\{t\,|\,
\widehat\theta^{(t)}\in\Theta_{\omega^+}\}$$ and put
$v'=1/\omega_{i_0}^2$ and $v''=1/\omega_{i_0}^{+2}$. 
From Lemma  \ref{lemma: moto parallelo}, and applying Lemma \ref{quellodellenote} to $\mu^{(t)}$ and $\nu^{(t)}$, 
we get that for
$t\in I'$ sufficiently large, it holds

\begin{equation}\label{reprfinal}
\widehat{\theta}^{(t+1)}_{i_0}-\widehat{\theta}^{(t)}_{i_0}=
c^{(t)}\gamma^{(t)}(\bar
y_\omega-\overline{\widehat{\theta}}^{(t)})+\frac{1}{\bar\nu^{(t)}}\gamma^{(t)}a^{(t)}_\delta
\left(v'-v''\right)(
y_{i_0}-\overline{\widehat{\theta}}^{(t)})+a^{(t)}_\delta+r^{(t)}.
\end{equation}
If ${\widehat\theta}(\omega)>\Theta_{\omega}\cap\Lambda$,
then also, by Lemma \ref{lemma:y_i0},
$\bar{y}_{\omega^+}>\Theta_{\omega}\cap\Lambda$. This, using
(\ref{parallel}), would imply that $\widehat{\theta}^{(t)}$ would
necessarily exit $\mathcal{L}_{\omega,\omega^+}$ in finite time.
Therefore, we must have
${\widehat\theta}(\omega)<\Theta_{\omega}\cap\Lambda$. Hence, $\bar
y_\omega-\overline{\widehat{\theta}}^{(t)}<0$. Moreover, it is
easy to check that in any case $\left(v'-v''\right)(
y_{i_0}-\overline{\widehat{\theta}}^{(t)})<0$. Recall now the
definition of the constant $\tilde c$ in (\ref{boundary}) and
notice that, since $\widehat{\theta}^{(t)}\in
\mathcal{L}_{\omega,\omega^+}$, $$c^{(t)}\gamma^{(t)}(\bar
y_\omega-\overline{\widehat{\theta}}^{(t)})\leq -\alpha^2\tilde
c/4\beta^2\gamma^{(t)}.$$ Choose now $\delta$ such that $\delta <
\alpha^2\tilde c/16\beta^2$ and $\bar t\geq \bar t_\delta$ such
that $r(t)<\delta\gamma^{(t)}$. It then follows from
(\ref{reprfinal}) that for $t\in I'$ and $t\geq \bar t$, it holds
$$\widehat{\theta}^{(t+1)}_{i_0}-\widehat{\theta}^{(t)}_{i_0}\leq -\alpha^2\tilde
c/8\beta^2\gamma^{(t)}<0.$$ This says that as long as
$\widehat{\theta}^{(t)}\in\Theta_{\omega}$, its $i_0$-th component
decreases. But this entails that $\widehat{\theta}^{(t)}$ can
never leave $\Theta_{\omega}$, which contradicts the infinite
switching assumption and thus implies the thesis.
\end{proof}

\subsection{Proof of Theorem \ref{teo:convergenza}}
Propositions \ref{inv},
\ref{inv2}, and \ref{limbo1} imply that there exists $\widehat\omega^{IA}\in\{\alpha,\beta\}^{\mathcal V}$ such that $\widehat\theta^{(t)}\in \Theta_{\widehat\omega^{IA}}$ for $t$ sufficiently large. This immediately implies that $\widehat\omega^{(t)}=\widehat\omega^{IA}$ for $t$ sufficiently large.
Corollary \ref{convergenza_finale} implies that $\widehat{\theta}^{IA}=\lim_{t\rightarrow+\infty}\widehat{\theta}^{(t)}=\hat\theta(\widehat\omega^{IA})$
Finally, since $\hat\theta(\widehat\omega^{IA})\in \Theta_{\widehat\omega^{IA}}$, we also have that $\widehat\omega^{IA}=\hat\omega(\widehat{\theta}^{IA})$.



\section{Proof of concentration results}\label{appB}
\subsection{Preliminaries}
For a more efficient parametrization of the stationary points, we introduce the notation: 
\begin{equation}\label{y omega}\omega\in\{\alpha,\beta\}^{\mathcal V}\quad  \Theta_{\omega}:=\{x\in\R\,|\, |x-y_i|<\delta\, \Leftrightarrow\, \omega_i=\alpha\}\end{equation}
It is then straightforward to check from (\ref{stationary1}) that the set of local maxima ${\mathcal S}_N$ can be represented as
\begin{equation}\label{stationary2} 
{\mathcal S}_N:=\{\theta =\widehat\theta(\omega)\,|\, \omega\in\{\alpha,\beta\}^{\mathcal V},\; \widehat\theta(\omega)\in\Theta_{\omega}\}.
\end{equation}
Since,
\begin{equation}\label{not empty}\Theta_{\omega}\neq \emptyset\;\Leftrightarrow\; \omega=\widehat \omega(x)\,\hbox{\rm for some}\; x\in\R\end{equation}
for analysing the set ${\mathcal S}_N$ we can restrict to consider $\omega$ of type $\omega=\widehat \omega(x)$. Consider the sequence of random functions
$\gamma_N(x):= \widehat\theta(\widehat \omega(x))$.

From \eqref{partial max}, applying the strong law of large numbers, we immediately get that
\begin{equation}\label{y inf}\lim_{N\rightarrow+\infty}\gamma_N(x)\stackrel{\mathrm{a.s.}}{=}\gamma_{\infty}(x):=\frac{\mathbb{E}(y_1\widehat \omega(x)_1^{-2})}{\mathbb{E}( \widehat \omega(x)_1^{-2})}.\end{equation}
Something stronger can indeed be said by a standard use of Chernoff bound \cite{Borkar}:
\begin{lemma}\label{lemma1}  For every $\epsilon >0$, there exists $q<1$ such that, for any $x\in\R$,
\begin{equation*}\mathbb{P}\left(\left|\gamma_N(x)-\gamma_{\infty}(x)\right|>\epsilon\right)\leq 2q^N.\end{equation*}
\end{lemma}

\begin{proof}
Let $a_i=y_i \omega_N(x)_i^{-2}$ and $b_i= \omega_N(x)_i^{-2}$ with $i\in\{1,\ldots,N\}$ and let $a$ and $b$ denote the corresponding expected values.

By Chernoff's bound and by Hoeffding's inequality we have, respectively, that
\begin{align*}&\mathbb{P}\left(\left|\frac{1}{N}\sum _{i=1}^N a_i-a\right|\geq\epsilon_1\right)\leq q_1^N\qquad\mathbb{P}\left(\left|\frac{1}{N}\sum _{i=1}^N b_i-b\right|\geq\epsilon_2\right)\leq 2 q_2^N\end{align*}
with \begin{equation}\label{q1q2}q_1=e^{-\frac{\alpha^2\epsilon_1^2}{4}}\qquad q_2=e^{-2\epsilon_2^2\left(\alpha^{-2}-\beta^{-2}\right)^{-2}}.\end{equation}

Fix $\epsilon_1<\frac{\epsilon}{2b\beta^4}$ and $\epsilon_2<\frac{\epsilon}{2|a|\beta^4}$, then
\begin{align*}
\mathbb{P}\left(\left|\bar{y}_{\omega_N(x)}-y_{\infty}(x)\right|>\epsilon\right)&\leq\mathbb{P}\left(\frac{\left|\frac{1}{N}\sum _{i=1}^N a_i-a\right|b+|a|\left|b-\frac{1}{N}\sum _{i=1}^N b_i\right|}{b\frac{1}{N}\sum _{i=1}^N b_i}>\epsilon \right)\\
%
&\leq q_1^N+q_2^N+\mathbf{1}_{\left\{\beta^4({\epsilon_1 b+|a|\epsilon_2})>\epsilon \right\}}\\
&=q_1^N+q_2^N
\end{align*}
where the last step follows by the way $\epsilon_1$ and $\epsilon_2$ have been chosen. 

There is still a point to be understood: in our derivation $q_1$ and $q_2$ depend on the choice of $x$ through $a$ and $b$. However, it is immediate to check that $a$ and $b$ are both bounded in $x$. This allows to conclude.
\end{proof}

From \eqref{y inf} is immediate to see that $\gamma_{\infty}$ is a bounded function of class $C^1$ and
it has an important property which will be useful later on.
\begin{lemma}\label{lemma1bis} There exists a constant $C>0$ such that
\begin{align*} &x-\gamma_{\infty}(x)\geq C(x-\theta^*)\;\quad {\rm if}\, x\in(\theta^{\star},+\infty) \\ &\gamma_{\infty}(x)-x\geq C(\theta^*-x)\; \quad{\rm if}\, x\in(-\infty,\theta^{\star})\\
&\gamma_{\infty}(\theta^*)=\theta^*
\end{align*}
\end{lemma}
\begin{proof}

If  $x\in(\theta^{\star},+\infty)$ and $f$ is the density of each $y_i$ (a mixture of  two Gaussians) then
\begin{align*}x-y_{\infty}(x)&=\frac{\frac{1}{\alpha^2}\int_{x-\delta}^{x+\delta}{(x-t)f(t)}\mathrm{d}t+ \frac{1}{\beta^2} \int_{\mathbb{R}\setminus(x-\delta,x+\delta)}{(x-t)f(t)}\mathrm{d}t}{\frac{1}{\alpha^2}\int_{x-\delta}^{x+\delta}{f(t)}\mathrm{d}t+ \frac{1}{\beta^2}\int_{\mathbb{R}\setminus(x-\delta,x+\delta)}{f(t)}\mathrm{d}t}\\
&\geq \frac{\frac{1}{\beta^2} \int_{\mathbb{R}}{(x-t)f(t)}\mathrm{d}t}{\frac{1}{\alpha^2}\int_{x-\delta}^{x+\delta}{f(t)}\mathrm{d}t+ \frac{1}{\beta^2}\int_{\mathbb{R}\setminus(x-\delta,x+\delta)}{f(t)}\mathrm{d}t}
\end{align*}
where the last inequality follows from the fact that $\int_{x-\delta}^{x+\delta}{(x-t)f(t)}\mathrm{d}t \geq 0$. We conclude that
\begin{align*}x-y_{\infty}(x)&\geq \frac{\frac{1}{\beta^2} (x-\theta^{\star})}{\frac{1}{\alpha^2}\int_{x-\delta}^{x+\delta}{f(t)}\mathrm{d}t+ \frac{1}{\beta^2}\int_{\mathbb{R}\setminus(x-\delta,x+\delta)}{f(t)}\mathrm{d}t}>0.
\end{align*}
Second statement if $x\in(-\infty,\theta^{\star})$ can be verified in a completely analogous way. The third statement then simply follows by continuity.
\end{proof}

We now come to a key result.
\begin{lemma}\label{lemma2}
For any fixed $\epsilon >0$, there exist $\tilde q\in(0,1)$ and $\chi>0$ such that
\begin{equation}
\mathbb{P}\left(\gamma_N(x)\in\Theta_{\widehat \omega(x)}\right)\leq \chi\tilde q^N
\end{equation}
for all $x$ such that $|x-\theta^{\star}|>\epsilon$.
\end{lemma}

\begin{proof}
We assume $x>\theta^{\star}+\epsilon$ (the other case $x<\theta^{\star}-\epsilon$ being completely equivalent).
Fix $\epsilon'\in (0, C\epsilon)$ where $C$ was defined in Lemma \ref{lemma1bis} and estimate as follows
\begin{equation}\label{estimbelong}\begin{split}\mathbb{P}\left(\gamma_N(x)\in\Theta_{\widehat \omega(x)}\right)&\leq
\mathbb{P}\left(\gamma_N(x)\in\Theta_{\widehat \omega(x)}\,,\; |\gamma_N(x)-\gamma_{\infty}(x)|\leq \epsilon'\right)\\&+
\mathbb{P}\left( |\gamma_N(x)-\gamma_{\infty}(x)|>\epsilon '\right).
\end{split}\end{equation}
Using Lemma \ref{lemma1bis} we get
\begin{equation*}\left\{|\gamma_N(x)-\gamma_{\infty}(x)|\leq \epsilon'\right\}\subseteq \left\{\gamma_N(x)\leq x-(C\epsilon-\epsilon')\right\}.\end{equation*} Thus
\begin{equation*}\begin{split}&\left\{\gamma_N(x)\in\Theta_{\widehat \omega(x)}, |\gamma_N(x)-\gamma_{\infty}(x)|\leq \epsilon'\right\}\\
&\qquad\subseteq
 \{\nexists i\,:\, y_i\in (\gamma_N(x)-\delta , \gamma_N(x)-\delta+\min\{C\epsilon-\epsilon', \delta\})\}\end{split}\end{equation*} and, consequently, the first term in \eqref{estimbelong} can be estimated as
 \begin{equation}\label{estimbelong2}
\begin{split}
\mathbb{P}&\left(\gamma_N(x)\in\Theta_{\widehat \omega_N(x)}\,,\; |\gamma_N(x)-\gamma_{\infty}(x)|\leq \epsilon'\right)\leq 
\left(1-\int_{\gamma_N(x)-\delta}^{\gamma_N(x)-\delta+\min\{C\epsilon-\epsilon', \delta\}}f(y)dy\right)^N
\end{split}\end{equation}
where $f(y)$ is the density of each $y_i$. Considering now that $f(y)$ is bounded away from $0$ on any bounded interval, that $|\gamma_N(x) -\gamma_{\infty}(x)|\leq \epsilon'$ and that $\gamma_{\infty}(x)$ is a bounded function, we deduce that the right hand side of (\ref{estimbelong2}) can be uniformly bounded as $\tilde q^N$ for some
$\tilde q\in (0,1)$. 
Substituting in (\ref{estimbelong}), and using Lemma \ref{lemma1} we finally obtain the thesis.
\end{proof}

\subsection{Proof of Theorem \ref{concentration2}}
 Define $$\mathcal{A}_N(\epsilon):=\left\{\exists \omega\in\{\alpha,\beta\}^{\mathcal{V}}: \widehat \theta(\omega) \in \Theta_{\omega}, |\widehat \theta(\omega) -\theta^{\star}|>\epsilon\right\}$$ for any $\epsilon>0$ and 
\begin{align*}
\mathcal{B}_1&:=\left\{\exists i\in\mathcal{V}: |y_i-\theta^{\star}|>N\right\}
\\
\mathcal{B}_2&:=\left\{\exists (i,j)\in\mathcal{V}\times{\mathcal{V}}: |y_i-y_j|<N^{-4}\right\}\\
\mathcal{B}_3&:=\{\exists (i,j)\in\mathcal{V}\times{\mathcal{V}}: |y_i-y_j|\in\left(2\delta,2\delta+{N^{-4}}\right\}\end{align*}
and estimate $\mathbb{P}\left(\mathcal{A}_N(\epsilon) \right)\leq \mathbb{P}\left(\mathcal{A}_N(\epsilon)  ,\mathcal{B}_1^c\cap\mathcal{B}_2^c \cap\mathcal{B}_3^c\right)+ \mathbb{P}(\mathcal{B}_1)+ \mathbb{P}(\mathcal{B}_2)+ \mathbb{P}(\mathcal{B}_3)$.
Standard considerations allow to upper bound the probability of each event $\mathcal{B}_i$ by a common term $K/N^2$.
We now focus on the estimation of the first term. 
The crucial point is that,  the condition  $\mathcal{B}_1^c\cap\mathcal{B}_2^c \cap\mathcal{B}_3^c$ allow us to reinforce condition (\ref{not empty}) in the sense that all $\omega$ for which $\Theta_\omega\neq\emptyset$ can be obtained as $\omega=\widehat \omega(x)$ as  $x$ varies in a set whose cardinality is polynomial in $N$. 
Specifically, define
 $$Z=\{\zeta_j= \theta^{\star}-N-\delta+{j}{N^{-4}}: j\in\mathbb{N}, j\leq j_{\rm max}\}$$ where $j_{\rm max}:=\lceil N^4(2N+2\delta)\rceil$
 and notice that, assuming that the $y_i$'s satisfy $\mathcal{B}_2^c \cap\mathcal{B}_3^c$, we have that $\widehat \omega(\zeta_j)$ and
$\widehat \omega(\zeta_{j+1})$ differ in at most one component and that $\widehat \omega(x)\in \{\widehat \omega(\zeta_j), \widehat \omega(\zeta_{j+1})\}$ for every $x\in
[\zeta_j, \zeta_{j+1}]$. Moreover, because of $\mathcal{B}_1^c$ we have that
 $\widehat \omega(x)_i=\widehat \omega(\zeta_0)_i=\beta$ for all $x\leq \theta^{\star}_N-\delta$ and for all $i$. Similarly, $\widehat \omega(x)_i=\widehat \omega(\zeta_{j_{\rm max}})_i=\beta$ for all $x\geq \theta^{\star}+N+\delta$ and for sll $i$. In other terms,
 under the assumption  that the $y_i$'s satisfy $\mathcal{B}_1^c\cap\mathcal{B}_2^c \cap\mathcal{B}_3^c$, it holds
 $\{\omega\in\{\alpha, \beta\}^{\mathcal V}\;|\; \Theta_{\omega}\neq\emptyset\}=\{\widehat \omega(x)\;|\; x\in Z\}$.
 Hence,
 \begin{align*}\mathbb{P}&\left(\mathcal{A}_N(\epsilon) , \mathcal{B}_1^c\cap\mathcal{B}_2^c \cap\mathcal{B}_3^c\right)\leq \\ &\leq \mathbb{P}\left(\bigcup_{\zeta\in Z}\left\{\gamma_N(\zeta)\in\Theta_{\widehat \omega(\zeta)} ,|\gamma_N(\zeta)-\theta^{\star}|>\epsilon\right\}\right)\\
 &\leq 
 \mathbb{P}\left(\bigcup_{\zeta\in Z}\left\{\gamma_N(\zeta)\in\Theta_{\widehat \omega(\zeta)} ,|\gamma_{\infty}(\zeta)-\theta^{\star}|>\epsilon/2\right\}\right)
 +\\&+\mathbb{P}\left(|\gamma_N(\zeta)-\gamma_{\infty}(\zeta)|\leq\epsilon/2\right).
 \end{align*}
 Notice that, because of the continuity of $\gamma_{\infty}$, there exists $\tilde\epsilon >0$ such that $|\gamma_{\infty}(\zeta)-\theta^{\star}|>\epsilon/2\;\Rightarrow\; |\zeta -\theta|>\tilde \epsilon$. We can then use Lemma \ref{lemma2}, 
 \begin{align*}&\mathbb{P}\left(\bigcup_{\zeta\in Z}\left\{\gamma_N(\zeta)\in\Theta_{\widehat \omega(\zeta)} ,|\gamma_{\infty}(\zeta)-\theta^{\star}|>\epsilon/2\right\}\right)\nonumber\\
 &\qquad\leq 
 |Z| \mathbb{P}\left(\gamma_N(\zeta)\in\Theta_{\widehat \omega(\zeta)} ,|\gamma_N(\zeta)-\theta^{\star}|>\epsilon\right) 
\leq cN^5\tilde q^N \end{align*}
where $c$ and $\tilde q$ are those coming from Lemma \ref{lemma2} relatively to $\tilde\epsilon$. Putting together all the estimations we have obtained and using Lemma \ref{lemma1}, we finally obtain that there exists $\chi>0$ such that
$\mathbb{P}\left(\mathcal{A}_N(\epsilon) \right)\leq \chi/N^2
$.
Using Borel-Cantelli Lemma and standard arguments, it follows now that the relation (\ref{local minima}) hold in an almost surely sense.

 It remains to be shown convergence in mean square sense. For this we need to go back to the form (\ref{derivative}) of the derivative of $L(\theta, \widehat \omega(\theta))$. The key observation is that the second additive term in the right hand side of (\ref{derivative}) 
can be bounded uniformly in modulus by some constant $C$. If we denote $\bar \gamma_N=N^{-1}\sum_iy_i$,
this implies that the function is increasing for $\theta>\bar \gamma_N+\beta^2C$ and decreasing for $\theta<\bar \gamma_N-\beta^2C$. Hence, necessarily,
\begin{equation}\label{bound}|\xi-\bar \gamma_N|\leq \beta^2C\;\ \forall\xi\in{\mathcal S}_N.\end{equation} 
On the other hand, by the law of large numbers, $\bar \gamma_N$ almost surely converges to $\theta^{\star}$ and this implies, by the previous part of the theorem that 
$\max\limits_{\xi\in{\mathcal S}_N}| \xi  -\bar \gamma_N|$ converges to $0$. This, together with \eqref{bound}, yields
$\E\max\limits_{\xi\in{\mathcal S}_N}| \xi  -\bar \gamma_N|^2\to 0$ for $N\to +\infty$. Since by the ergodic theorem also
$\E|\bar \gamma_N-\theta^{\star}|^2\to 0$ for $N\to +\infty$, the proof is complete.

\subsection{Proof of Proposition \ref{concentration4}}
We prove it for $\widehat \omega^{\mathrm{IA}}$, the other verification being completely equivalent). If $\sigma\in\{\alpha, \beta\}$, we define
\begin{equation*}\begin{split}
 f(\theta,\sigma)& =\P(\widehat \omega(\theta)_i\neq\sigma\;|\; \omega^{\star}_i=\sigma)\\ &=\begin{cases}\frac{1}{\sqrt{2\pi \sigma^2}}\int_{{\theta}-\delta}^{{\theta} +\delta}\mathrm{e}^{-\frac{(s-\theta^{\star})^2}{2\sigma^2}}\mathrm{d} s \qquad \ \ \text{if }\sigma=\beta
 \\
 1-\frac{1}{\sqrt{2\pi \sigma^2}}\int_{{\theta}-\delta}^{{\theta} +\delta}\mathrm{e}^{-\frac{(s-\theta^{\star})^2}{2\sigma^2}}\mathrm{d} s \quad \text{if }\sigma=\alpha
 \end{cases}\end{split}\end{equation*}
 (notice that $f$ does not depend on $i$). 
 We can compute
 \begin{equation*}\begin{split}\frac{1}{N}\E d_H(\widehat \omega^{\mathrm{IA}}, \omega^{\star})&=\frac{1}{N}\sum\limits_i\P(\widehat \omega^{\mathrm{IA}}_i\neq \omega^{\star}_i)\\&=p\E f(\widehat\theta^{\mathrm{IA}},\alpha)+(1-p)
 \E f(\widehat\theta^{\mathrm{IA}},\beta).\end{split}\end{equation*}
 Since $f(\theta,\sigma)$ is a $C^1$ function of $\theta$, we immediately obtain that
 $$|\E f(\widehat\theta^{\mathrm{IA}},\sigma)-\E f(\theta^{\star},\sigma)|\leq C\E|\widehat\theta^{\mathrm{IA}}-\theta^{\star}|$$
 and, by Corollary \ref{concentration3}, this last expression converges to $0$, for $N\to +\infty$. Hence,
  $$\frac{1}{N}\E d_H(\widehat \omega^{\mathrm{IA}}, \omega^{\star})=p\E f(\theta^{\star},\alpha)+(1-p)\E f(\theta^{\star},\beta).$$
 Straightforward computation now proves the thesis.
 \end{document}